\documentclass[11pt]{article}

\usepackage[a4paper,margin=1in]{geometry}
\usepackage{amsmath,amssymb,amsfonts,mathtools,bm,dsfont} 
\usepackage{amsthm}
\usepackage{booktabs}
\usepackage{longtable}
\usepackage{siunitx}
\usepackage{enumitem}
\usepackage[hyphens]{url} 
\usepackage{hyperref}
\hypersetup{colorlinks=true,allcolors=blue,breaklinks=true} 
\usepackage[authoryear,round]{natbib}
\usepackage{algorithm, algpseudocode} 
\usepackage{caption}
\usepackage{graphicx}
\usepackage{xcolor} 
\usepackage{tabularx}
\usepackage{rotating}
\usepackage[normalem]{ulem}  

\definecolor{zhangcolor}{rgb}{0.8, 0.1, 0.1} 
\definecolor{xucolor}{rgb}{0.1, 0.6, 0.1}    
\definecolor{qucolor}{rgb}{0.8, 0.5, 0.1}    

\newcommand{\E}{\mathbb{E}}
\newcommand{\X}{\bm{X}}

\newcommand{\iid}{\stackrel{\mathrm{i.i.d.}}{\sim}}
\newcommand{\pto}{\overset{p}{\to}}
\newcommand{\dto}{\xrightarrow{d}}

\newif\iffastdemo
\fastdemofalse
\newcommand{\DemoPrefix}{\textbf{FAST-DEMO PLACEHOLDER---NOT FOR SCIENTIFIC INTERPRETATION OR SUBMISSION.}}

\newcommand{\ResultPlaceholder}[1]{\fbox{\parbox{0.92\linewidth}{\centering\small #1}}}
\newcommand{\DemoFigure}[4]{\begin{figure}[htbp]\centering
\IfFileExists{#1}{\includegraphics[width=#2\textwidth]{#1}}{\ResultPlaceholder{Expected generated asset: \path{#1}}}
\caption{\DemoPrefix\ #3}\label{#4}\end{figure}}

\newcommand{\DemoTwoFigure}[6]{\begin{figure}[htbp]\centering
\begin{minipage}[t]{0.49\textwidth}\centering\IfFileExists{#1}{\includegraphics[width=\linewidth]{#1}}{\ResultPlaceholder{Expected generated asset: \path{#1}}}\\[-2pt]{\small (a) #2}\end{minipage}\hfill
\begin{minipage}[t]{0.49\textwidth}\centering\IfFileExists{#3}{\includegraphics[width=\linewidth]{#3}}{\ResultPlaceholder{Expected generated asset: \path{#3}}}\\[-2pt]{\small (b) #4}\end{minipage}
\caption{\DemoPrefix\ #5}\label{#6}\end{figure}}
\newcommand{\DemoStackedTwoFigure}[6]{\begin{figure}[p]\centering
\IfFileExists{#1}{\includegraphics[width=0.70\textwidth]{#1}}{\ResultPlaceholder{Expected generated asset: \path{#1}}}\\[-2pt]{\small (a) #2}\\[5pt]
\IfFileExists{#3}{\includegraphics[width=0.82\textwidth]{#3}}{\ResultPlaceholder{Expected generated asset: \path{#3}}}\\[-2pt]{\small (b) #4}
\caption{\DemoPrefix\ #5}\label{#6}\end{figure}}
\newcommand{\DemoThreeFigure}[8]{\begin{figure}[htbp]\centering
\IfFileExists{#1}{\includegraphics[width=0.82\textwidth]{#1}}{\ResultPlaceholder{Expected asset: \path{#1}}}\\[-2pt]{\small (a) #2}\\[5pt]
\begin{minipage}[t]{0.48\textwidth}\centering\IfFileExists{#3}{\includegraphics[width=\linewidth]{#3}}{\ResultPlaceholder{Expected asset: \path{#3}}}\\[-2pt]{\small (b) #4}\end{minipage}\hfill
\begin{minipage}[t]{0.48\textwidth}\centering\IfFileExists{#5}{\includegraphics[width=\linewidth]{#5}}{\ResultPlaceholder{Expected asset: \path{#5}}}\\[-2pt]{\small (c) #6}\end{minipage}
\caption{\DemoPrefix\ #7}\label{#8}\end{figure}}

\newtheorem{assumption}{Assumption}
\newtheorem{theorem}{Theorem}

\newtheorem{lemma}{Lemma}

\newtheorem{corollary}{Corollary}
\newtheorem{proposition}{Proposition}

\title{\textbf{Uniform-Design Subsampling for Compute-Budgeted Double Machine Learning}}
\author{
     Yuanke Qu$^{1,\#}$, Xiaoya Xu$^{2,\#}$, 
    and Hengtao Zhang$^{1,*}$\\[4pt]
    \small $^{1}$School of Computer Science and Engineering, Guangdong Ocean University, Guangdong 529500, China\\
    \small $^{2}$Institute of Applied Mathematics, Shenzhen Polytechnic University, Shenzhen 518055, China\\
    \small $^\#$\,Co-first authors. These authors contributed equally to this work. \\
    \small $^*$\,Corresponding author: \texttt{zhanght@gdou.edu.cn}
}
\date{\today}

\begin{document}
\maketitle

\begin{abstract}
Double machine learning (DML) combines orthogonal scores with flexible nuisance estimation, but repeated cross-fitting and repeated analysis can make full-data workflows expensive. When a fixed computational budget requires a working sample of size $r\ll n$, simple random subsampling may cover the covariate space poorly and produce unstable treated--control composition. We propose Uniform Design Double Machine Learning (UD-DML), which maps a common low-discrepancy skeleton through empirical marginal quantiles and matches each anchor to one treated and one control observation, without replacement within each treatment arm. We establish finite-sample integration and balance bounds, an exact selected-target decomposition, and a selection-aware leading-score central limit theorem on the root-$r$ scale. Under explicit moment, weighted-calibration, variance-stabilisation, and target-transport conditions, a residual-based variance estimator consistently studentises the original UD-DML estimator. Full-scale simulations show that the common-skeleton construction provides its clearest improvements over uniform subsampling under non-uniform covariate geometry and limited treated--control overlap, while maintaining empirical coverage near the nominal level. These results provide a principled working-sample design for repeated causal learning under an explicit computational budget.
\end{abstract}

\vspace{1em}

\section{Introduction}

Large observational databases do not remove computational constraints; they make the unit of analysis a complete workflow rather than a single model fit. Electronic health records, insurance claims, platform logs, and linked administrative data can contain millions of observations \citep{stuart_dugoff_abrams_salkever_steinwachs_2013,chetty_hendren_2018,luo_zhuang_xie_zhu_wang_an_xu_2024}. A causal analysis on such data commonly repeats treatment and outcome learning across cross-fitting folds, tuning configurations, bootstrap samples, and sensitivity specifications \citep{zivich_breskin_2021,kosko_wang_santacatterina_2024}. Even when one full-data fit is possible, the full workflow may exceed a wall-clock, memory, or cloud-computing budget. The literature on optimal and design-inspired subsampling treats this constraint as a statistical design problem: once only a fixed number of observations can enter the expensive fitting stage, the retained observations should be chosen deliberately rather than regarded as an arbitrary computational shortcut \citep{Yao2021_Review,Yu2024_Review}.

This issue is especially consequential for double/debiased machine learning (DML), which combines Neyman-orthogonal scores with cross-fitting to estimate low-dimensional causal parameters while allowing flexible nuisance learners \citep{bickel1993efficient,newey1994asymptotic,BangRobins2005,Chernozhukov2018EJ}. For the average treatment effect, every fold requires estimation of the propensity score and treatment-specific outcome regressions, followed by evaluation on held-out observations. We study the regime in which the repeated nuisance-fitting stage has a working-sample budget $r\ll n$. The statistical design problem is to choose the $r$ observations used for nuisance fitting and score evaluation while preserving covariate coverage and treated--control alignment.

That decision cannot be reduced to replacing $n$ by $r$ in the DML algorithm. Simple random subsampling may represent common regions of the covariate distribution adequately while omitting sparse but consequential regions, and its realised treatment composition can be unstable when overlap is limited. The selected treated and control observations may then occupy different parts of the covariate space, aggravating nuisance-estimation error and the variability of inverse-probability contributions \citep{Rosenbaum1983_PS,Lunceford2004}. Moreover, treatment-dependent selection changes the empirical law on which the orthogonal score is evaluated. A fixed-budget DML method must therefore address covariate coverage, cross-arm alignment, and transport from the selected working sample to the original target.

Existing subsampling methods offer two broad design principles. Model-based optimal subsampling assigns inclusion probabilities or selects observations to optimise an information or variance criterion for a specified model and estimator \citep{Yao2021_Review,Yu2024_Review}. This approach has been developed for linear and generalised linear regression \citep{wang2019information,AiYuZhangWang2021_Sinica}, logistic regression \citep{WangZhuMa2018_JASA,Wang2019_MoreEfficient}, quantile regression \citep{WangMa2021_Biometrika}, additive hazards models \citep{Zuo2021_AHM}, and quasi-likelihood estimation \citep{Yu2022_DistQuasi}. Its efficiency criterion is valuable when the working model is credible, but it is less natural when several nuisance functions are learned flexibly and no single parametric model determines which observations are informative. Geometric subsampling instead seeks a space-filling representation of the observed covariates. Uniform-design and quasi-Monte Carlo methods provide tools for this purpose \citep{fang2000uniform,Niederreiter1992,Dick2013}.

The closest geometric method is the model-free UD subsampler of \citet{zhang2023model}, which combines dimension reduction, marginal empirical transformations, a power-generator good-lattice-point design selected by mixture discrepancy, and nearest-observation replacement. We retain these geometric components and redesign the sampling object for causal estimation. UD-DML constructs one common skeleton and assigns to every anchor one treated and one control observation, producing two treatment-specific empirical samples coupled through shared reference locations. This coupling enforces equal allocation and allows the skeleton contribution to cancel in the treated--control balance bound. Cross-fitted DML is then fitted to the selected original observations. Related causal subsampling methods use staged nuisance estimation and calibration \citep{Su2022_CSDA,Su2026_CSDA}. In particular, \citet{Su2026_CSDA} estimates complex nuisance functions on a subsample but computes the final estimator using the full sample. In contrast, UD-DML restricts both nuisance estimation and score evaluation to $r$ observations.

The common-skeleton construction also guides the theory. PCA organizes the retained covariates into orthogonal directions, and the empirical transformations place those directions on comparable marginal scales. The resulting skeleton supplies shared reference locations for treated and control selection. Realised generalized empirical $F$-discrepancy and matching radii then quantify how accurately the matched observations approximate the transformed empirical covariate distribution. These quantities yield finite-sample integration and arm-balance bounds and enter the rate condition for transport from the selected sample to the original target. Under explicit moment, calibration, variance-stabilisation, and rate conditions, the leading score of the original UD-DML estimator has a selection-conditional root-$r$ Gaussian limit, which yields the stated unconditional asymptotic normality.

The paper makes three connected contributions. It formulates treatment-paired, common-skeleton selection as a causal working-sample design under a fixed fitting budget. It provides design-conditional finite-sample guarantees and a selection-aware decomposition that identifies the representation and transport quantities governing inference after subsampling. It also derives a selection-conditional leading-score limit and resulting asymptotic normality with a matching residual-based variance estimator, and evaluates the method through simulation and a large observational application. All working-sample methods are compared at the same $r$-budget, with FULL-DML included as a computational and statistical reference where available.

Section~\ref{sec:methodology} develops the working-sample construction, algorithm, computational cost, and inferential guarantees. Section~\ref{sec:simulations} studies accuracy, uncertainty, scalability, and the components of the sampling design. Section~\ref{sec:real_data} applies the method to US natality records, and Section~\ref{sec:conclusion} discusses its scope and limitations.

\section{Methodology}
\label{sec:methodology}

\subsection{DML under a Fixed Working-Sample Budget}\label{sec:preliminaries}

Consider $n$ i.i.d.\ observations $\bm{O}_i=(Y_i,W_i,\bm{X}_i)$, $i=1,\ldots,n$, where $Y_i\in\mathbb{R}$ is the observed outcome, $W_i\in\{0,1\}$ is the treatment indicator, and $\bm{X}_i\in\mathbb{R}^p$ contains pre-treatment covariates. Under the potential-outcomes formulation \citep{Rubin1974}, $Y_i=W_iY_i(1)+(1-W_i)Y_i(0)$ and the target is the average treatment effect $\theta_0={E}\{Y(1)-Y(0)\}$. We assume unconfoundedness, $(Y(1),Y(0))\perp W\mid\bm{X}$, and overlap, $0<{P}(W=1\mid\bm{X})<1$ almost surely. These conditions identify $\theta_0$, but they do not determine how a restricted working sample should be constructed when nuisance fitting cannot use all $n$ observations. That design choice is the starting point of UD-DML.

The DML framework of \citet{Chernozhukov2018EJ} represents the ATE by the orthogonal moment condition $\mathbb{E}\{\psi(\bm{O};\theta_0,\bm{\eta}_0)\}=0$, where $\bm{\eta}_0$ contains the conditional outcome regressions $m_g(\bm{x})={E}[Y\mid\bm{X}=\bm{x},W=g]$, $g\in\{0,1\}$, and the propensity score $e(\bm{x})={P}(W=1\mid\bm{X}=\bm{x})$. We use the AIPW score
\begin{equation}\label{eq:aipw}
\psi_{\text{AIPW}}(\bm{O};\theta,\bm{\eta})=\bigl(m_1(\bm{X})-m_0(\bm{X})\bigr)+\frac{W\{Y-m_1(\bm{X})\}}{e(\bm{X})}-\frac{(1-W)\{Y-m_0(\bm{X})\}}{1-e(\bm{X})}-\theta.
\end{equation}
Here $\widehat{\bm{\eta}}=(\widehat m_0,\widehat m_1,\widehat e)$ denotes the nuisance estimates. Under $K$-fold cross-fitting, $\widehat{\bm{\eta}}^{(-k)}$ is trained outside fold $k$ and evaluated inside that fold. The resulting estimator solves
\[
\frac{1}{n}\sum_{i=1}^n \psi\!\left(\bm{O}_i;\widehat{\theta},\widehat{\bm{\eta}}^{(-k_i)}\right)=0,
\]
where $k_i$ denotes the fold containing observation $i$. Orthogonality controls the first-order effect of nuisance-estimation error, and the AIPW score is doubly robust in the usual full-data setting \citep{BangRobins2005}. Under the corresponding nuisance-rate and regularity conditions, full-data DML admits root-$n$ inference.

Data-dependent subsampling places nuisance learning under a new empirical covariate law and can alter the treatment composition and the target represented by the selected score. UD-DML addresses these effects at the design stage. It selects both treatment arms jointly through shared anchors, quantifies the realised selected-to-full empirical discrepancy, and decomposes transport from the selected target to the original
population target under an additional condition. The resulting inferential scale is $\sqrt r$, which links precision directly to the working-sample budget and supports principled allocation of computation across repeated analyses.

\subsection{Common-Skeleton Causal Subsampling}\label{sec:UDsub}
The design object in UD-DML is not a single representative subset. It is a pair of arm-specific empirical samples coupled through the same $r_p$ anchors. Each anchor is assigned one treated and one control observation, so the final budget $r=2r_p$ is equally allocated and both arms are asked to cover the same locations in the transformed covariate space. The common skeleton therefore provides the shared reference measure against which treated--control imbalance can be bounded. Distributional fidelity is an observed property rather than a consequence of the name ``uniform design'' and depends on the realised discrepancy and the arm-specific matching radii defined below.

The geometric ingredients used to construct the anchors follow the model-free UD subsampling procedure of \citet{zhang2023model}: PCA rotation, marginal empirical transformation, a power-generator good-lattice-point design selected by mixture discrepancy, and nearest-observation replacement. Our contribution begins with a different statistical object and a different use of those anchors. The method of \citet{zhang2023model} constructs one representative sample for generic downstream learning, whereas UD-DML uses one skeleton to couple two treatment-specific selections and then evaluates an orthogonal causal score on the paired working sample. When the two selected arm averages are subtracted, their common skeleton average cancels; two independently constructed skeletons do not have this cancellation. This coupling yields exact equal allocation, the treated--control balance bound in Section~\ref{sec:theory_inference}, and a selection problem that must be accounted for in the root-$r$ analysis.

Let $\bm X_i \in \mathbb R^p$ denote the $p$-dimensional covariate vector for unit $i$, $i=1,\ldots,n$, and let $\bm X=(\bm X_1^\top,\ldots,\bm X_n^\top)^\top$ be the corresponding $n\times p$ covariate matrix. We construct the anchors in PCA coordinates. The rotation produces orthogonal empirical directions, and truncation concentrates the matching problem on the directions accounting for a prescribed share of total covariate variation. Let
\[
n_g=\sum_{i=1}^n\mathds{1}\{W_i=g\},\qquad g\in\{0,1\},
\qquad
1\leq r_p\leq\min(n_0,n_1).
\]
We assume that every input covariate has positive sample standard deviation. Let $\bar{\bm X}=n^{-1}\sum_{i=1}^n \bm X_i$ be the sample mean of the covariates, and let $\widehat{\bm D}$ be the resulting positive diagonal matrix of sample standard deviations. Thus $\widehat{\bm D}^{-1}$ exists and the standardized data have positive total PCA energy. We standardize the covariates by
\[
\widetilde{\bm X}_i=\widehat{\bm D}^{-1}(\bm X_i-\bar{\bm X}), \qquad i=1,\ldots,n.
\]
Write the singular value decomposition of the standardized matrix as $\widetilde{\bm X}=\bm U \bm \Sigma \bm V^\top$, where $\bm \Sigma=\mathrm{diag}(\sigma_1,\ldots,\sigma_p)$ with $\sigma_1\ge \cdots \ge \sigma_p\ge 0$. Let $\rho_0\in(0,1)$ be a prescribed cumulative variance threshold, and define $q$ to be the smallest integer such that
\(
{(\sum_{d=1}^q \sigma_d^2)}/{(\sum_{d=1}^p \sigma_d^2)}\ge \rho_0.
\)
In our experiments we use $\rho_0=0.85$, following earlier implementations \cite{zhang2023model,Zhou2024_Techno}. This is a working default rather than an outcome-tuned optimum. In applications, $\rho_0$ should be chosen without using $Y$, subject to a maximum feasible $q$, and checked using retained variance, achieved GEFD, and matching radii; low-variance directions may still be outcome-relevant. Let $\bm V_q$ be the first $q$ columns of $\bm V$. We then define the retained rotated covariates by
\[
\bm Z_i=\bm V_q^\top \widetilde{\bm X}_i \in \mathbb R^q, \qquad i=1,\ldots,n.
\]
Here $q\le p$ is the working dimension used in the UD construction. When $q=p$, no dimension reduction is performed. 

The next step constructs the shared reference locations. For a prescribed pair budget $r_p$, an $r_p$-run, $q$-factor uniform design is a point set in $[0,1]^q$ chosen to have low discrepancy from the product-uniform reference distribution. The points are later mapped to the empirical marginal scales of the rotated covariates and reused for both treatment arms. Following \citet{zhang2023model}, we generate candidate skeletons by the leave-one-out good lattice point method with a power generator. Let $\gcd(\cdot,\cdot)$ denote the greatest common divisor.
A positive integer $\gamma$ is said to be admissible if $\gcd(\gamma,r_p+1)=1$ and the $q$ residues $1,\gamma,\gamma^2,\ldots,\gamma^{q-1}\pmod{r_p+1}$ are mutually distinct. For an admissible $\gamma$, define the power generator vector by
\[
\bm \gamma=\bigl(\gamma^0,\gamma^1,\ldots,\gamma^{q-1}\bigr)^\top.
\]
The corresponding candidate $r_p$-run $q$-factor design is
\[
\mathcal{U}_{r_p}(\gamma)=\bigl(\bm u_j^{(\gamma)}\bigr)_{j=1}^{r_p}\subset[0,1]^q,
\]
where
\[
\bm u_j^{(\gamma)}
=
\frac{\mathrm{mod}(j\bm \gamma,\,r_p+1)}{r_p}
-\frac{1}{2r_p}\bm 1_q,
\qquad j=1,\ldots,r_p,
\]
$\bm 1_q$ is the $q$-dimensional vector of ones, and $\mathrm{mod}(\cdot,\cdot)$ is applied componentwise. Different admissible values of $\gamma$ give different candidate skeletons. We select among them by a criterion that measures coverage of the reference cube; this criterion governs anchor construction and is not itself a claim that the observed joint covariate distribution is uniform.

We use the mixture discrepancy, a generalized $L_2$-discrepancy criterion for space-filling designs \citep{zhou2013mixture,zhang2023model}. Let $\mathcal D=(\bm u_j)_{j=1}^{r_p}\subset[0,1]^q$ be an indexed candidate design, and let $F_\mathcal D$ and $F_u$ denote the empirical distribution of $\mathcal D$ and the uniform distribution on $[0,1]^q$, respectively. For a reproducing kernel $\mathcal{K}:[0,1]^q\times[0,1]^q\to\mathbb R$, define the squared discrepancy of $\mathcal D$ from $F_u$ by
\begin{align*}
D^2(\mathcal D;F_u,\mathcal{K})
&=
\int_{[0,1]^{2q}}
\mathcal{K}(\bm u,\bm t)\,d\{F_\mathcal D(\bm u)-F_u(\bm u)\}\,d\{F_\mathcal D(\bm t)-F_u(\bm t)\} \\
&=
\int_{[0,1]^{2q}} \mathcal{K}(\bm u,\bm t)\,d\bm u\,d\bm t
-\frac{2}{r_p}\sum_{j=1}^{r_p}\int_{[0,1]^q} \mathcal{K}(\bm u,\bm u_j)\,d\bm u \\
&\quad
+\frac{1}{r_p^2}\sum_{j=1}^{r_p}\sum_{k=1}^{r_p} \mathcal{K}(\bm u_j,\bm u_k),
\end{align*}
where the second equality follows by expanding the quadratic form. A smaller $D^2(\mathcal D;F_u,\mathcal{K})$ indicates closer coverage of the reference cube. For the mixture-discrepancy kernel,
\[
\mathcal{K}_M(\bm u,\bm t)
=
\prod_{d=1}^q
\left(
\frac{15}{8}
-\frac{1}{4}\left|u_d-\frac{1}{2}\right|
-\frac{1}{4}\left|t_d-\frac{1}{2}\right|
-\frac{3}{4}|u_d-t_d|
+\frac{1}{2}|u_d-t_d|^2
\right),
\quad \bm u,\bm t\in[0,1]^q.
\]
The corresponding discrepancy
$
D_M(\mathcal D):=D(\mathcal D;F_u,\mathcal{K}_M)
$
is called the mixture discrepancy. Since $\mathcal{K}_M$ is a product kernel, the above integrals can be evaluated coordinatewise, which yields the closed-form expression
\begin{align*}
D_M^2(\mathcal D)
&=
\left(\frac{19}{12}\right)^q
-\frac{2}{r_p}
\sum_{j=1}^{r_p}
\prod_{d=1}^q
\left[
\frac{5}{3}
-\frac{1}{4}\left|u_{jd}-\frac{1}{2}\right|
-\frac{1}{4}\left(u_{jd}-\frac{1}{2}\right)^2
\right] \\
&\quad
+\frac{1}{r_p^2}
\sum_{j=1}^{r_p}\sum_{k=1}^{r_p}
\prod_{d=1}^q
\left(
\frac{15}{8}
-\frac{1}{4}\left|u_{jd}-\frac{1}{2}\right|
-\frac{1}{4}\left|u_{kd}-\frac{1}{2}\right|
-\frac{3}{4}|u_{jd}-u_{kd}|
+\frac{1}{2}|u_{jd}-u_{kd}|^2
\right).
\end{align*}
The derivation is given in Appendix~\ref{sec:MD_deriv}. We choose the skeleton by minimizing this criterion over admissible generators. Let $\mathcal A_{r_p,q}$ denote the set of admissible values of $\gamma$. We assume $\mathcal A_{r_p,q}\neq\varnothing$ and fix a search budget $B_\gamma\in\mathbb N_+$. Exhaustive enumeration of $\mathcal A_{r_p,q}$ can be unnecessarily expensive when $r_p$ is moderately large.
We thus implement the quasi-optimal search over a budgeted random subset of size $B_\gamma$, 
\(
\mathcal{G}^{(B_\gamma)}_{r_p,q} \subset \mathcal{A}_{r_p,q},
\) with $|\mathcal{G}^{(B_\gamma)}_{r_p,q}| = \min\{B_\gamma, |\mathcal{A}_{r_p,q}|\}$.
The searched candidate collection is therefore nonempty. We define
\[
\widehat{\gamma}\in\arg\min_{\gamma\in\mathcal{G}^{(B_\gamma)}_{r_p,q}} D_M^2\!\bigl(\mathcal U_{r_p}(\gamma)\bigr),
\]
choosing the smallest numerical generator if the minimum is attained more than once, and set $\mathcal U_{r_p}=\mathcal U_{r_p}(\widehat{\gamma})=(\bm u_j)_{j=1}^{r_p}$.
This produces a budgeted search for a space-filling skeleton relative to $F_u$. The value $B_\gamma=30$ used below is a computational default. A practical alternative is to stop when additional candidate generators yield negligible improvement in mixture discrepancy.

The reference-cube skeleton must next be expressed on the scale at which observed units are matched. For each coordinate $d=1,\ldots,q$, let $\widehat F_{Z^{(d)}}$ denote the empirical cumulative distribution function of $\{Z_i^{(d)}\}_{i=1}^n$, and define the marginal empirical transformation
\[
T_{\bm Z}(\bm z)
=
\bigl(
\widehat F_{Z^{(1)}}(z^{(1)}),\ldots,\widehat F_{Z^{(q)}}(z^{(q)})
\bigr)^\top.
\]
The corresponding skeleton points in the rotated space are
\[
\bm v_j
=
T_{\bm Z}^{-1}(\bm u_j)
=
\Bigl(
\widehat F_{Z^{(1)}}^{-1}(u_{j1}),\ldots,\widehat F_{Z^{(q)}}^{-1}(u_{jq})
\Bigr)^\top,
\qquad j=1,\ldots,r_p.
\]
Each $\bm v_j\in\mathbb R^q$ is an anchor in the PCA-rotated covariate space. The coordinatewise map $T_{\bm Z}$ calibrates only the empirical marginal scales. It does not transform the joint empirical distribution
into the product-uniform distribution. Mixture discrepancy selects a space-filling reference skeleton, and the realised empirical discrepancy and matching radii measure how closely its matched anchors follow the observed joint support.

The methodological distinction from applying UD separately within each arm appears at the matching step. The same anchor $\bm v_j$ is used once in the treated arm and once in the control arm, and matching is performed without replacement within each arm. Thus, the two selected empirical measures are linked anchor by anchor rather than merely having the same sample size. The preceding scale, energy, arm-size, and generator conditions make the inverse standardisation, retained PCA coordinates, skeleton search, and without-replacement matching well defined. Nearest-observation argmin ties are resolved by selecting the smallest original observation index among equal-distance matches. With $\mathcal U_{g,j-1}=\{i_1^g,\ldots,i_{j-1}^g\}$, for $j=1,\ldots,r_p$ we set
$$i_j^g=\arg\min_{i:W_i=g,\ i\notin\mathcal U_{g,j-1}}\|\bm{Z}_i-\bm{v}_j\|_2,\qquad g\in\{0,1\},$$
with the stated smallest-index tie rule.
Distances are computed in the retained rotated space, but the anchors generate no synthetic outcomes or covariates. Once $i_j^1$ and $i_j^0$ are obtained, estimation uses the original observations $(Y_{i_j^g},W_{i_j^g},\bm X_{i_j^g})$, $g\in\{0,1\}$. Define $\mathcal S_1=\{i_j^1\}_{j=1}^{r_p}$, $\mathcal S_0=\{i_j^0\}_{j=1}^{r_p}$, and $\mathcal S=\mathcal S_1\cup\mathcal S_0$. The selected sample has size $r=2r_p$ and is
\(
\bigl\{(Y_i,W_i,\bm X_i):i\in\mathcal S\bigr\}.
\) 
The resulting sample retains original outcomes and covariates while coupling the two treatment-specific empirical measures through a common geometric reference. Section~\ref{sec:theory_inference} quantifies its realised covariate representation and arm alignment using generalized empirical $F$-discrepancy and arm-specific matching radii.

\subsection{Algorithm, Tuning, and Computational Cost}\label{sec:UD-DML Algo}
Algorithm~\ref{alg:UD-DML} translates the common-skeleton construction into three computational stages. The first constructs the anchors and selects distinct treated--control pairs in the retained PCA space. The second fits cross-fitted DML to the selected observations on their original scale. The third computes the point estimate and the selection-aware residual variance studied in Section~\ref{sec:theory_inference}. The algorithm is called only when the preceding feasibility conditions hold; the deterministic tie rule is part of the selection map. A Python implementation is available at \url{https://github.com/BobZhangHT/UD-DML}.

\begin{algorithm}[htbp]
\caption{Uniform Design Double Machine Learning (UD-DML)}
\label{alg:UD-DML}
\footnotesize
\begin{algorithmic}[1]
\Require Full data $\{(Y_i,W_i,\bm X_i)\}_{i=1}^n$ satisfying positive sample standard deviations, positive total standardized PCA energy, $1\le r_p\le\min(n_0,n_1)$ with $n_g=\sum_i\mathds{1}\{W_i=g\}$, and nonempty $\mathcal A_{r_p,q}$; retained-variance threshold $\rho_0$; $K$ folds for DML; quasi-optimal search budget $B_\gamma\in\mathbb N_+$.
\Ensure ATE estimate $\widehat{\theta}_{\mathrm{UD}}$, residual variance estimate $\widehat{V}_{\mathrm{UD}}$, and a Wald interval.
\State Let $r=2r_p$.

\Statex \textbf{Phase 1: UD subsampling in the retained PCA space}
\State Standardize the covariates:
\(
\widetilde{\bm X}_i=\widehat{\bm D}^{-1}(\bm X_i-\bar{\bm X}), ~~i=1,\ldots,n.
\)
\State Compute the singular value decomposition $\widetilde{\bm X}=\bm{U}\bm{\Sigma} \bm{V}^\top$, and choose the smallest $q$ such that
\(
{(\sum_{d=1}^q \sigma_d^2)}/{(\sum_{d=1}^p \sigma_d^2)}\ge \rho_0.
\)
\State Form the retained rotated covariates
\(
\bm Z_i=\bm V_q^\top \widetilde{\bm X}_i \in \mathbb R^q,~~ i=1,\ldots,n.
\)
\State Compute the marginal empirical CDFs $\hat{F}_{Z^{(d)}}(\cdot)$ of $\{Z_i^{(d)}\}_{i=1}^n$ for $d=1,\ldots,q$.
\If{a cached skeleton exists for the triple $(r_p,q,B_\gamma)$}
    \State Load the cached design $\mathcal U_{r_p}=\{\bm u_j\}_{j=1}^{r_p}$.
\Else
    \State Randomly draw a candidate subset $\mathcal{G}^{(B_\gamma)}_{r_p,q} \subset \mathcal{A}_{r_p,q}$ with
    \(
    |\mathcal{G}^{(B_\gamma)}_{r_p,q}| = \min\!\{B_\gamma,\ |\mathcal{A}_{r_p,q}|\}.
    \)
    \For{each candidate $\gamma \in \mathcal{G}^{(B_\gamma)}_{r_p,q}$}
        \State Construct the $r_p$-run $q$-factor candidate design $\mathcal U_{r_p}(\gamma)=\{\bm u_j^{(\gamma)}\}_{j=1}^{r_p}\subset[0,1]^q$.
        \State Compute its mixture discrepancy $D_M^2\!\left(\mathcal U_{r_p}(\gamma)\right)$.
    \EndFor
    \State Choose
    \[
     \widehat{\gamma}=\min\arg\min_{\gamma\in \mathcal{G}^{(B_\gamma)}_{r_p,q}}
    D_M^2\!\left(\mathcal U_{r_p}(\gamma)\right),
    \]
    and set $\mathcal U_{r_p}=\mathcal{U}_{r_p}(\hat{\gamma})=\{\bm u_j\}_{j=1}^{r_p}$.
    \State Cache $\mathcal U_{r_p}$ for future repeated runs with the same $(r_p,q,B_\gamma)$.
\EndIf
\State Map the design points back to the retained rotated space by the empirical inverse transformation:
\[
\bm v_j=T_{\bm Z}^{-1}(\bm u_j)
=
\Bigl(
\widehat F_{Z^{(1)}}^{-1}(u_{j1}),\ldots,\widehat F_{Z^{(q)}}^{-1}(u_{jq})
\Bigr)^\top,\qquad j=1,\ldots,r_p.
\]
\State Build \texttt{cKDTree} on the rotated covariates:
\(\mathcal I_1 \text{ on } \{\bm Z_i:W_i=1\},\) and 
\(\mathcal I_0 \text{ on } \{\bm Z_i:W_i=0\}.\)
\State Initialize $\mathcal S_1=\emptyset$, $\mathcal S_0=\emptyset$ and arm-specific used-index sets $\mathcal U_1=\mathcal U_0=\emptyset$.
\For{$j=1,\ldots,r_p$}
\State Find the smallest-index minimizer
    \(
    i_j^1=\arg\min_{i:W_i=1,\ i\notin\mathcal U_1}\|\bm Z_i-\bm v_j\|_2
    \)
    by querying $\mathcal I_1$ and enlarging the candidate list until an unused index is found.
\State Find the smallest-index minimizer
    \(
    i_j^0=\arg\min_{i:W_i=0,\ i\notin\mathcal U_0}\|\bm Z_i-\bm v_j\|_2
    \)
    by querying $\mathcal I_0$ and enlarging the candidate list until an unused index is found.
    \State Update
    \(\mathcal S_1\leftarrow \mathcal S_1\cup\{i_j^1\},\)
    \(\mathcal S_0\leftarrow \mathcal S_0\cup\{i_j^0\},\)
    $\mathcal U_1\leftarrow\mathcal U_1\cup\{i_j^1\}$, and
    $\mathcal U_0\leftarrow\mathcal U_0\cup\{i_j^0\}$.
\EndFor
\State Set $\mathcal S=\mathcal S_1\cup\mathcal S_0$.

\Statex \textbf{Phase 2: Cross-fitted DML on the selected original observations}
\State Run the standard $K$-fold cross-fitting DML procedure on
\(
\{(Y_i,W_i,\bm X_i):i\in\mathcal S\}.
\)
\State Obtain the cross-fitted nuisance estimators $\widehat{\bm\eta}^{(-k)}$ and the pseudo-outcomes
\(
\widehat\psi_i^*
=
\psi^*\!\left(\bm{O}_i; \widehat{\bm\eta}^{(-k_i)}\right),~i\in\mathcal S.
\)

\Statex \textbf{Phase 3: Estimation and selection-aware uncertainty}
\State Compute
\(
\widehat{\theta}_{\mathrm{UD}}=\frac{1}{r}\sum_{i\in\mathcal S}\widehat\psi_i^*.
\)
\State For every $i\in\mathcal S$, compute the cross-fitted residual contribution
\(
\widehat\xi_i=\frac{W_i(Y_i-\widehat m_1^{(-k_i)}(\bm X_i))}{\widehat e_i}
-\frac{(1-W_i)(Y_i-\widehat m_0^{(-k_i)}(\bm X_i))}{1-\widehat e_i}
\)
and set
\(
\widehat{V}_{\mathrm{UD}}=r^{-2}\sum_{i\in\mathcal S}\widehat\xi_i^2.
\)
\State Construct the $(1-\alpha)$ Wald interval
\(
\widehat{\theta}_{\mathrm{UD}}\pm z_{1-\alpha/2}\sqrt{\widehat{V}_{\mathrm{UD}}}.
\)
\end{algorithmic}
\end{algorithm}

Matching without replacement is essential to the statistical procedure. It guarantees $|\mathcal S_g|=r_p$ and prevents one original observation from entering more than one cross-fitting fold. Allowing duplicate nearest neighbours would instead produce a multiset, change the effective sample size, and invalidate the usual fold-independence and variance interpretation. The implementation therefore enforces one-to-one matching within each treatment arm and records the selected-sample size, number of unique indices, and arm counts for every replication.

All design parameters are chosen without using outcomes. The primary tuning quantity is $r$, which translates the available fitting budget into a statistical working-sample size. In practice, we recommend evaluating a short increasing sequence of $r$ values and stopping when the available computation is exhausted or when GEFD, mean and maximum matching radii, and covariate standardised mean differences show little further improvement. The threshold $\rho_0$ trades a lower matching dimension against the risk of omitting low-variance directions, and should therefore be assessed jointly with these diagnostics. The search budget $B_\gamma$ can be increased until the best mixture discrepancy stabilises. We use $K=2$ in the large simulation study and $K=5$ in the data application; within each experiment, all methods use the same fold count, nuisance learners, and propensity clipping. The implementation also records the achieved GEFD approximation, mean and maximum matching radii, SMD summaries, component runtimes, and process peak memory, so that a favourable point estimate cannot substitute for evidence that the intended design was actually obtained.

The cost separates into full-data preprocessing, skeleton search, matching, and nuisance fitting. Standardisation, PCA, and the marginal empirical transformations are performed once; after the rotated covariates are available, sorting the $q$ coordinates costs $O(qn\log n)$. Evaluation of one candidate skeleton costs $O(r_p^2q)$ through the closed-form mixture discrepancy, so a cold search over $B_\gamma$ candidates costs $O(B_\gamma r_p^2q)$. The skeleton can be cached whenever $(r_p,q,B_\gamma)$ is unchanged.

For matching, a direct scan of both arms at every anchor costs $O(r_pnq)$. The implementation instead builds one \texttt{cKDTree} per arm and enlarges each query until it finds the nearest unused observation. In low-to-moderate dimension, $O(r_pq\log n)$ describes the average query cost, although competition among anchors can increase it and this expression is not a worst-case guarantee. Once $\mathcal S$ has been selected, cross-fitting depends on $r$ rather than $n$; write its cost as $C_{\mathrm{DML}}(r,p,K)$. Computing the point estimate and residual variance adds $O(r)$.

Combining the above components, the total cold-start computational cost of UD-DML is
\[
C_{\mathrm{PCA}}(n,p,q)
+
O(qn\log n)
+
O(B_\gamma r_p^2 q)
+
O(r_p q\log n)
+
C_{\mathrm{DML}}(r,p,K),
\]
where $C_{\mathrm{PCA}}(n,p,q)$ denotes the one-time PCA preprocessing cost.
In repeated Monte Carlo experiments the UD skeleton depends only on $(r_p,q,B_\gamma)$; once the quasi-optimal design has been generated and cached, the amortised cost per replication reduces to
\[
C_{\mathrm{PCA}}(n,p,q)
+
O(qn\log n)
+
O(r_p q\log n)
+
C_{\mathrm{DML}}(r,p,K),
\]
up to negligible cache-lookup overhead. After caching, the iterative learner cost is $C_{\mathrm{DML}}(r,p,K)$. For comparison, standard full-data DML has leading cost
\[
C_{\mathrm{FULL}}(n,p,K)\asymp C_{\mathrm{DML}}(n,p,K),
\]
up to a lower-order inference term. UD-DML still processes all $n$ covariate vectors, but it confines repeated nuisance learning to the selected $r$ observations. This distinction is useful when nuisance fitting, rather than one-time geometric preprocessing, exhausts the workflow budget.

For illustration, suppose tree-ensemble training on $m$ observations has approximate cost
\[
c_{K,T}m\log m,
\]
where $c_{K,T}$ depends on the learner and cross-fitting configuration. Full-data fitting then uses $m=n$, whereas UD-DML uses $m=r$ after preprocessing. The realised saving depends on $(n,p,q,r)$, the learner, and cache reuse. Each reported LightGBM fit uses \texttt{n\_jobs=1}, and any concurrency is across independent Monte Carlo or bootstrap replications; the timings are therefore per-call benchmarks under this implementation.

\subsection{Design Guarantees and Root-$r$ Inference}\label{sec:theory_inference}

The theory connects three parts of the working-sample construction. The first quantifies how closely selected covariates reproduce full-sample empirical averages. The second controls treated--control imbalance through the common skeleton, while the corollary identifies the additional conditions needed for transport to the original target. The third establishes inference for the original target from the $r$ selected observations. Deterministic finite-sample bounds govern representation and balance, while stochastic conditions on outcome noise, nuisance estimation, conditional variance, and componentwise rates yield a selection-conditional leading-score limit and hence root-$r$ asymptotic normality for the estimator.

The representation analysis begins with the generalized empirical $F$-discrepancy (GEFD) of \citet{zhang2023model}. Let $\mathcal{Z}=(\bm{Z}_i)_{i=1}^n$ denote the retained PCA-rotated covariates, where $\bm{Z}_i=\bm{V}_q^\top \widehat{\bm{D}}^{-1}(\bm{X}_i-\bar{\bm{X}})$. For a finite indexed collection $\mathcal A=(\bm a_\ell)_{\ell=1}^{N_{\mathcal A}}\subset\mathbb R^q$, repetitions are retained and its empirical measure is
\[
\widehat P_{\mathcal A}=N_{\mathcal A}^{-1}\sum_{\ell=1}^{N_{\mathcal A}}\delta_{\bm a_\ell}.
\]
Thus, the notation records the observation index even when empirical quantiles or transformed anchors coincide. For another indexed collection $\mathcal P=(\bm \xi_k)_{k=1}^m\subset\mathbb R^q$, define the transformed collections
\[
\widetilde{\mathcal Z}=(\widetilde{\bm Z}_i)_{i=1}^n,
\quad \widetilde{\bm Z}_i=T_{\bm Z}(\bm Z_i),
\qquad
\widetilde{\mathcal P}=(\widetilde{\bm\xi}_k)_{k=1}^m,
\quad \widetilde{\bm\xi}_k=T_{\bm Z}(\bm\xi_k).
\]
Both transformed collections lie in $[0,1]^q$. For a reproducing kernel $\mathcal{K}:[0,1]^q\times[0,1]^q\to\mathbb R$, define their squared GEFD by
\begin{align}
D^2(\widetilde{\mathcal P};\widetilde{\mathcal Z},\mathcal{K})
&=
\frac{1}{n^2}\sum_{i=1}^n\sum_{i'=1}^n
\mathcal{K}\!\left(\widetilde{\bm Z}_i,\widetilde{\bm Z}_{i'}\right)
-\frac{2}{nm}\sum_{i=1}^n\sum_{k=1}^{m}
\mathcal{K}\!\left(\widetilde{\bm Z}_i,\widetilde{\bm\xi}_k\right) \notag\\
&\quad
+\frac{1}{m^2}\sum_{k=1}^{m}\sum_{k'=1}^{m}
\mathcal{K}\!\left(\widetilde{\bm\xi}_k,\widetilde{\bm\xi}_{k'}\right). \label{eq:gefd_sec24}
\end{align}
Here $T_{\bm Z}$ is the marginal empirical transformation defined in Section~\ref{sec:UDsub}. Whenever $D$ compares finite collections, its arguments are interpreted through their empirical measures, so repeated transformed observations remain in the averages. A term of the form $D(\mathcal A;F_u,\mathcal K)$ compares a transformed collection $\mathcal A\subset[0,1]^q$ with the product-uniform reference on the same space, as in the mixture discrepancy of Section~\ref{sec:UDsub}. The GEFD is conditional on the observed design and does not measure discrepancy from the population law. Let $\mathcal H_{\mathcal K}$ be the reproducing-kernel Hilbert space of a positive-definite kernel $\mathcal K$, and define the kernel variation by
\[
V_{\mathcal K}(f):=\|f\|_{\mathcal H_{\mathcal K}},\qquad f\in\mathcal H_{\mathcal K}.
\]
The following empirical RKHS Koksma--Hlawka lemma gives the form used below; its proof is included in Appendix~\ref{app:proof_ekh_rkhs}.
\begin{lemma}[Empirical RKHS Koksma--Hlawka inequality]\label{lem:ekh_rkhs}
For indexed collections $\mathcal A=(\bm a_\ell)_{\ell=1}^{N_{\mathcal A}}$ and $\mathcal B=(\bm b_s)_{s=1}^{N_{\mathcal B}}$ in $[0,1]^q$ and every $f\in\mathcal H_{\mathcal K}$,
\begin{equation}\label{eq:ekh_sec24}
\left|
\frac{1}{N_{\mathcal A}}\sum_{\ell=1}^{N_{\mathcal A}} f(\bm a_\ell)
-
\frac{1}{N_{\mathcal B}}\sum_{s=1}^{N_{\mathcal B}} f(\bm b_s)
\right|
\le
D(\mathcal A;\mathcal B,\mathcal K)\,V_{\mathcal K}(f).
\end{equation}
\end{lemma}
This formulation uses only the reproducing property and Cauchy--Schwarz; the cited works provide the empirical-discrepancy and kernel-integration background for this representation \citep{zhang2023model,Hickernell1998}.
Equation~\eqref{eq:ekh_sec24} controls two empirical averages in the transformed retained-coordinate space. The coordinatewise transformation standardises the empirical marginal scales, while Proposition~\ref{prop:copula_decomp} relates the realised point-set GEFD to the searched mixture discrepancy, the empirical inverse-transform approximation, and the joint geometry of the transformed sample.

\begin{proposition}[GEFD decomposition]\label{prop:copula_decomp}
Let $\mathcal V_{r_p}=(\bm v_j)_{j=1}^{r_p}\subset\mathbb R^q$ be the indexed UD skeleton and let $\widetilde{\mathcal V}_{r_p}=(T_{\bm Z}(\bm v_j))_{j=1}^{r_p}\subset[0,1]^q$ be its transformed collection. Let $\mathcal U_{r_p}=(\bm u_j)_{j=1}^{r_p}$ and let $F_u$ be the product-uniform distribution on $[0,1]^q$. For any positive-definite kernel $\mathcal K$ for which the corresponding kernel mean embeddings exist,
\begin{equation}\label{eq:copula_decomp}
\begin{aligned}
D(\widetilde{\mathcal V}_{r_p};\widetilde{\mathcal Z},\mathcal K)
&\leq D(\widetilde{\mathcal V}_{r_p};F_u,\mathcal K)
+D(\widetilde{\mathcal Z};F_u,\mathcal K)\\
&\leq D(\widetilde{\mathcal V}_{r_p};\mathcal U_{r_p},\mathcal K)
+D(\mathcal U_{r_p};F_u,\mathcal K)
+D(\widetilde{\mathcal Z};F_u,\mathcal K).
\end{aligned}
\end{equation}
For $\mathcal K=\mathcal K_M$, $D(\mathcal U_{r_p};F_u,\mathcal K_M)=D_M(\mathcal U_{r_p})$ is the mixture discrepancy used to search over skeletons. The other two terms record the empirical inverse-transform approximation and the joint geometry of the transformed sample.
\end{proposition}
The proof in Appendix~\ref{app:proof_copula_decomp} applies the triangle inequality to the point-set kernel mean embeddings in Proposition~\ref{prop:copula_decomp}. The left-hand side uses precisely the point-set GEFD defined in~\eqref{eq:gefd_sec24}. Thus the searched mixture discrepancy, empirical inverse-transform approximation, and transformed-sample geometry jointly bound the realised GEFD used after matching.

The anchors are geometric reference points, whereas DML uses the original observations indexed by $\mathcal S_1$, $\mathcal S_0$, and $\mathcal S$. Passing from anchors to matched observations therefore introduces a second approximation component. Define the PCA projection map
\(
\Pi_q(\bm x)
=
\bm V_q^\top \widehat{\bm D}^{-1}(\bm x-\bar{\bm X}),~
\bm x\in\mathbb R^p,
\)
and the induced transformation
\(
\mathcal T(\bm x)
=
T_{\bm Z}\!\bigl(\Pi_q(\bm x)\bigr)\in[0,1]^q.
\)
For each $g\in\{0,1\}$, the selected indices satisfy
\[
\mathcal S_g=\{i_j^g\}_{j=1}^{r_p},
\qquad
i_j^g=\arg\min_{i:W_i=g,\ i\notin\mathcal U_{g,j-1}}
\|\bm Z_i-\bm v_j\|_2,
\]
where $\mathcal U_{g,j-1}=\{i_1^g,\ldots,i_{j-1}^g\}$, as in the without-replacement algorithm in Section \ref{sec:UDsub}. Here $\mathcal V_{r_p}=(\bm v_j)_{j=1}^{r_p}\subset\mathbb R^q$ is the indexed UD skeleton. We also define the group-specific matching radius in the rotated space,
\[
\delta_g^{(\mathrm{rot})}
:=
\max_{1\le j\le r_p}
\|\bm Z_{i_j^g}-\bm v_j\|_2,
\qquad g\in\{0,1\}.
\]
The observable radius $\delta_g^{(\mathrm{rot})}$ is the largest displacement required to replace an anchor by a distinct observation from arm $g$. The following result separates skeleton representation error from matching error for functions of the retained PCA coordinates.

\begin{theorem}[Finite-sample integration bound]\label{thm:rep_ud}
Let $\mathcal V_{r_p}=(\bm v_j)_{j=1}^{r_p}\subset\mathbb R^q$ be the indexed UD skeleton, and let $\mathcal{K}$ be a positive-definite reproducing kernel on $[0,1]^q\times[0,1]^q$. Suppose $f\in\mathcal H_{\mathcal K}$, so that $V_{\mathcal K}(f)=\|f\|_{\mathcal H_{\mathcal K}}<\infty$, and define the induced function on the original covariate space by
\[
\phi_f(\bm x)
=
f\!\bigl(\mathcal T(\bm x)\bigr)
=
f\!\left(T_{\bm Z}\!\bigl(\Pi_q(\bm x)\bigr)\right).
\]
Assume moreover that the function
\[
\varphi_f(\bm z)
=
f\!\left(T_{\bm Z}(\bm z)\right),
\qquad \bm z\in\mathbb R^q,
\]
is Lipschitz continuous on the support of $\mathcal Z\cup\mathcal V_{r_p}$ with Lipschitz constant $L_f$, namely,
\[
|\varphi_f(\bm z)-\varphi_f(\bm z')|
\le
L_f\|\bm z-\bm z'\|_2
\qquad
\text{for all }\bm z,\bm z'\in \mathrm{supp}(\mathcal Z\cup\mathcal V_{r_p}).
\]
Then, for each $g\in\{0,1\}$,
\begin{equation}\label{eq:thm1_group}
\left|
\int \phi_f(\bm x)\,dP_{n,\bm X}(\bm x)
-
\int \phi_f(\bm x)\,dP_{\mathcal S_g,\bm X}(\bm x)
\right|
\le
D(\widetilde{\mathcal V}_{r_p};\widetilde{\mathcal Z},\mathcal{K})\,V_{\mathcal K}(f)
+
L_f\,\delta_g^{(\mathrm{rot})}.
\end{equation}
Consequently, for the total subsample empirical measure
\[
P_{\mathcal S,\bm X}
=
\frac{1}{2}P_{\mathcal S_1,\bm X}
+
\frac{1}{2}P_{\mathcal S_0,\bm X},
\]
we have
\begin{equation}\label{eq:thm1_total}
\left|
\int \phi_f(\bm x)\,dP_{n,\bm X}(\bm x)
-
\int \phi_f(\bm x)\,dP_{\mathcal S,\bm X}(\bm x)
\right|
\le
D(\widetilde{\mathcal V}_{r_p};\widetilde{\mathcal Z},\mathcal{K})\,V_{\mathcal K}(f)
+
\frac{L_f}{2}
\bigl(
\delta_1^{(\mathrm{rot})}
+
\delta_0^{(\mathrm{rot})}
\bigr).
\end{equation}
\end{theorem}

Theorem~\ref{thm:rep_ud} is directly evaluable for each realised working sample because its right-hand side contains observed discrepancy and matching-radius terms. For every admissible $\phi_f(\bm{x})=f\{\mathcal{T}(\bm{x})\}$, it bounds the selected-to-full empirical difference by a skeleton term and an arm-specific matching term. The choices $f(\bm u)=u_d$ and $f(\bm u)=u_du_\ell$ cover first- and second-order functionals of the transformed retained coordinates under the stated variation condition. Outcome regressions, propensity scores, and treatment-effect functions enter the same result whenever their dependence on the retained coordinates satisfies the stated smoothness condition. Thus the theorem supplies finite-sample representation bounds for the functions used downstream in causal estimation.

The same construction gives a sharper arm-to-arm statement. Comparing a treated observation and a control observation through their shared anchor cancels the skeleton average, leaving only the two matching errors and the regularity of the balanced function.
\begin{theorem}[Finite-sample arm-balance bound]\label{thm:bal_ud}
Let the conditions of Theorem~\ref{thm:rep_ud} hold. Then, for the induced function
\[
\phi_f(\bm x)=f\!\bigl(\mathcal T(\bm x)\bigr),
\]
the empirical measures of the treated and control UD-selected subsamples satisfy
\begin{equation}\label{eq:thm2_main}
\left|
\int \phi_f(\bm x)\,dP_{\mathcal S_1,\bm X}(\bm x)
-
\int \phi_f(\bm x)\,dP_{\mathcal S_0,\bm X}(\bm x)
\right|
\le
L_f\bigl(\delta_1^{(\mathrm{rot})}+\delta_0^{(\mathrm{rot})}\bigr).
\end{equation}
\end{theorem}

Theorem~\ref{thm:bal_ud} isolates the common-reference effect beyond equal arm sizes. It applies to functions of the retained transformed coordinates and links their treated--control imbalance to the two observable matching radii. When the downstream nuisance and treatment-effect functions belong to the admissible class, the integration and balance bounds provide the design quantities used in the selected-score analysis.

The design bounds do not by themselves identify the estimand represented by DML after deterministic, treatment-dependent selection. Although the rule does not use $Y$, it changes the marginal distribution of $(\bm X,W)$ presented to the nuisance learners. We therefore separate estimation under the selected law from transport to the original population. Define the AIPW pseudo-outcome
$$\psi^\ast(\bm{O};\bm{\eta}) := \psi_{\text{AIPW}}(\bm{O};0,\bm{\eta})=\bigl\{m_1(\bm{X})-m_0(\bm{X})\bigr\}+\frac{W\{Y-m_1(\bm{X})\}}{e(\bm{X})}-\frac{(1-W)\{Y-m_0(\bm{X})\}}{1-e(\bm{X})}.$$
Let $P$ denote the original population law. Consider repeated draws of a dataset of \(n\) observations, application of the selection rule with budget $r$, and uniform selection of one observation from the resulting working sample. Let $Q_{r,n}$ be the marginal law of that observation. Its propensity $e_{Q_{r,n}}(\bm x)=Q_{r,n}(W=1\mid\bm X=\bm x)$ generally differs from $e_P(\bm x)$ because the rule changes both covariate distribution and treatment composition. If selection uses only $(\bm X,W)$ and preserves the original conditional outcome regressions, the selected-law target is
\[
\theta_{Q_{r,n}}=E_{Q_{r,n}}\{m_{1,P}(\bm X)-m_{0,P}(\bm X)\},
\]
The implemented cross-fitted estimator on the UD sample is
\[
\widehat\theta_{\mathrm{UD}}
=
\frac{1}{r}\sum_{i\in\mathcal S}\widehat\psi_i^\ast,
\qquad
\widehat\psi_i^\ast
:=
\psi^\ast\!\bigl(\bm O_i;\widehat{\bm\eta}^{(-k_i)}\bigr),
\]
where $k_i$ is the fold label and $\widehat{\bm\eta}^{(-k_i)}$ is trained on the other selected observations. The propensity learner in this procedure targets the selected-law propensity, not $e_P$.

\begin{proposition}[Target decomposition]\label{prop:target_decomp}
For every realised estimator and every well-defined selected-law target,
\begin{equation}\label{eq:target_decomp}
\widehat\theta_{\mathrm{UD}}-\theta_0
=
\bigl(\widehat\theta_{\mathrm{UD}}-\theta_{Q_{r,n}}\bigr)
+
\bigl(\theta_{Q_{r,n}}-\theta_0\bigr).
\end{equation}
The first term is estimation error under the selected law. The second is the target shift caused by using the working sample. If $\tau_P(\bm x)=m_{1,P}(\bm x)-m_{0,P}(\bm x)$ has the form $f_\tau\{\mathcal T(\bm x)\}+b_\tau(\bm x)$ with $f_\tau\in\mathcal H_{\mathcal K}$, then Theorem~\ref{thm:rep_ud}, applied to $f_\tau$, bounds the empirical analogue of the target-shift term together with the difference between the selected and full-sample averages of $b_\tau$.
\end{proposition}
Although Proposition~\ref{prop:target_decomp} is an exact algebraic identity, it fixes the structure of the inferential argument. The first term concerns estimation under the law induced by selection, and the second concerns transport to the original population. Accurate selected-law estimation cannot compensate for a non-negligible target shift. The deterministic representation bound contributes to inference only when it controls that shift for the treatment-effect function. Appendix~\ref{app:proof_target_decomp} gives the proof.

The leading score is analysed by conditioning on the realised selection, fold assignment, covariates, treatments, and all outcome-independent auxiliary randomisation. Let $\mathcal R_n$ collect the auxiliary randomisation used by selection, fold assignment, propensity fitting, and generator search; for a deterministic implementation, $\mathcal R_n$ is degenerate. Define
\[
\mathcal G_n=\sigma\{\bm X_1,W_1,\ldots,\bm X_n,W_n,\mathcal R_n,\mathcal S,\text{fold labels}\}.
\]
Because UD selection and propensity fitting use only $(\bm X,W,\mathcal R_n)$, the selected indices and cross-fitted propensity predictions $\widehat e_i$ are $\mathcal G_n$-measurable. Define
\[
\theta_{\mathcal S,n}=\frac1r\sum_{i\in\mathcal S}\tau_P(\bm X_i),\qquad
\varepsilon_{g,i}=Y_i-m_{g,P}(\bm X_i),
\]
and $\xi_{i,n}=W_i\varepsilon_{1,i}/\widehat e_i-(1-W_i)\varepsilon_{0,i}/(1-\widehat e_i)$. Thus $\theta_{\mathcal S,n}$ is the realised selected-covariate target, while $\xi_{i,n}$ contains the outcome noise left after conditioning on the UD design.

\begin{assumption}[Conditions for selected-score inference]\label{as:ud_root_r}
Along a sequence with $r=r_n\to\infty$, the following conditions hold.

(i) The observations are i.i.d., and consistency and conditional unconfoundedness hold. Conditional on the full $(\bm X,W)$ array, $\mathcal R_n$ is independent of the outcomes. Selection is without replacement, uses only $(\bm X,W,\mathcal R_n)$, and returns exactly $r/2$ observations from each arm. The number of cross-fitting folds is fixed, and $r/n\to0$.

(ii) For constants $\epsilon>0$ and $\nu>0$, the fitted propensities and outcome residuals satisfy
\[
\epsilon\le\widehat e_i\le1-\epsilon\quad(i\in\mathcal S),
\qquad
\frac1r\sum_{i\in\mathcal S}
E(|\varepsilon_{W_i,i}|^{2+\nu}\mid\mathcal G_n)=O_P(1).
\]
Moreover,
\[
\sigma_{r,n}^2:=\frac1r\sum_{i\in\mathcal S}E(\xi_{i,n}^2\mid\mathcal G_n)
\pto\sigma_{\rm UD}^2
\]
for a constant $0<\sigma_{\rm UD}^2<\infty$.

(iii) Define
\[
\delta_{g,i}
=\widehat m_g^{(-k_i)}(\bm X_i)-m_{g,P}(\bm X_i).
\]
The selected empirical errors satisfy
\[
\frac1r\sum_{i\in\mathcal S}
\left(\delta_{0,i}^2+\delta_{1,i}^2\right)
=o_P(1)
\]
and
\begin{equation}\label{eq:ud_calibration}
B_{r,n}:=\frac1r\sum_{i\in\mathcal S}\left[
\delta_{1,i}\left(1-\frac{W_i}{\widehat e_i}\right)
+\delta_{0,i}\left\{\frac{1-W_i}{1-\widehat e_i}-1\right\}\right]
=o_P(r^{-1/2}).
\end{equation}

\end{assumption}

The three parts of Assumption~\ref{as:ud_root_r} correspond to
distinct components of the selected-score analysis. Part (i) makes
selection outcome-blind and preserves conditional independence of
outcome residuals at distinct selected indices. Part (ii) imposes
propensity boundedness, a Lyapunov-type moment condition, and
conditional variance stabilisation. Part (iii) imposes selected-sample
$L_2$ consistency of the outcome-regression estimators and a weighted
calibration rate at the selected evaluation points. This formulation
remains meaningful when the selected-law propensity differs from
$e_P$.

If $\widehat e_i=1/2$, the calibration remainder in
\eqref{eq:ud_calibration} becomes
\[
B_{r,n}
=
\frac12
(P_{\mathcal S_0}-P_{\mathcal S_1})
(\delta_0+\delta_1).
\]
This identity explains the connection between common-skeleton arm
balance and nuisance calibration. However, because the regression
errors are evaluated using fold-specific nuisance estimators,
Theorem~\ref{thm:bal_ud} does not by itself establish
\eqref{eq:ud_calibration} for the implemented cross-fitted estimator.
Accordingly, \eqref{eq:ud_calibration} is retained as a high-level
calibration condition. It must be verified for the learner and
asymptotic sequence under consideration and does not follow from
overlap or equal arm sizes alone.

The preceding conditions control estimation under the realised
selected design. It remains to connect the selected-covariate target
$\theta_{\mathcal S,n}$ to the original population target $\theta_0$.
The following corollary converts the finite-sample representation
bound into a sufficient condition for this transport step. Rather
than assuming the target shift directly, it separates the required
rate into covariate-representation, anchor-matching, and
residual-function components.

\begin{corollary}[Transport to the original target]
\label{cor:target_transport}
Suppose that
\[
\tau_P(\bm x)
=
f_\tau\{\mathcal T(\bm x)\}+b_\tau(\bm x),
\qquad
f_\tau\in\mathcal H_{\mathcal K},
\]
and define
\[
\Delta_{\tau,n}
=
D(\widetilde{\mathcal V}_{r_p};
  \widetilde{\mathcal Z},\mathcal K)
V_{\mathcal K}(f_\tau)
+
\frac{L_\tau}{2}
\{\delta_1^{(\mathrm{rot})}+\delta_0^{(\mathrm{rot})}\}
+
\left|
(P_{\mathcal S,\bm X}-P_{n,\bm X})b_\tau
\right|.
\]
If
\[
E\{\tau_P(\bm X)^2\}<\infty,\qquad
\sqrt r\,\Delta_{\tau,n}=o_P(1),
\qquad
r/n\to0,
\]
then
\[
\sqrt r(\theta_{\mathcal S,n}-\theta_0)=o_P(1).
\]
\end{corollary}
The condition $\sqrt r\,\Delta_{\tau,n}=o_P(1)$ separates three
sources of target-transport error. The first term is the GEFD
representation error, weighted by the RKHS variation of the component
of the treatment-effect function captured by the retained transformed
coordinates. The second term is the error incurred when the geometric
anchors are replaced by observed treated and control units. The third
term measures the selected-to-full empirical discrepancy of the
residual component $b_\tau$, which is not represented by the retained
transformed coordinates. Thus, target stability is not imposed as an independent assumption:
the corollary gives componentwise rate conditions under which the
design bounds imply root-$r$ transport to the original target. 

\begin{theorem}[Root-$r$ normality and studentisation under outcome-independent selection]\label{thm:ud_root_r}
Under Assumption~\ref{as:ud_root_r} and the
conditions of Corollary~\ref{cor:target_transport},
for every fixed $t\in\mathbb R$,
\begin{equation}\label{eq:ud_theorem_conditional_cf}
E\!\left[\left.\exp\!\left\{
\frac{it}{\sqrt r}\sum_{i\in\mathcal S}\xi_{i,n}
\right\}\,\right|\,\mathcal G_n\right]
\pto
\exp\!\left(-\frac{t^2\sigma_{\rm UD}^2}{2}\right).
\end{equation}
Consequently,
\begin{equation}\label{eq:ud_linear}
\sqrt r(\widehat\theta_{\mathrm{UD}}-\theta_0)
=\frac1{\sqrt r}\sum_{i\in\mathcal S}\xi_{i,n}+o_P(1)
\dto\mathcal N(0,\sigma_{\rm UD}^2).
\end{equation}
The conditional characteristic-function statement concerns the leading score array. The calibration remainder is controlled by Assumption~\ref{as:ud_root_r}(iii), whereas target transport follows from Corollary~\ref{cor:target_transport}.

Define
\[
\widehat\xi_i=
\frac{W_i\{Y_i-\widehat m_1^{(-k_i)}(\bm X_i)\}}{\widehat e_i}
-\frac{(1-W_i)\{Y_i-\widehat m_0^{(-k_i)}(\bm X_i)\}}{1-\widehat e_i}
\]
and
\begin{equation}\label{eq:ud_residual_variance}
\widehat\sigma_{\rm UD}^2=\frac1r\sum_{i\in\mathcal S}\widehat\xi_i^2,
\qquad \widehat V_{\rm UD}=\widehat\sigma_{\rm UD}^2/r.
\end{equation}
Then $\widehat\sigma_{\rm UD}^2\pto\sigma_{\rm UD}^2$ and
$(\widehat\theta_{\mathrm{UD}}-\theta_0)/\sqrt{\widehat V_{\rm UD}}\dto\mathcal N(0,1)$, so the associated Wald interval has limiting coverage $1-\alpha$.
\end{theorem}

Appendix~\ref{app:proof_ud_root_r} proves Theorem~\ref{thm:ud_root_r} by applying a conditional Lindeberg--Feller theorem to the triangular array $\{\xi_{i,n}:i\in\mathcal S\}$ and controlling the calibration and target-shift remainders. The variance estimator in \eqref{eq:ud_residual_variance} targets the conditional outcome-noise variance generated after the realised UD covariates have been fixed. Its studentised statistic therefore matches the resulting Gaussian limit. Assumption~\ref{as:ud_root_r} and Corollary~\ref{cor:target_transport} state the nuisance-learning, variance, and transport conditions that deliver asymptotic normality and Wald coverage at the $\sqrt r$ scale set by the computational budget.


\section{Simulated Studies}
\label{sec:simulations}
The simulation study uses an ordered four-stage ablation under a common working-sample budget. Section~\ref{sec:simulation_settings} defines the four sampling rules and the adjacent contrasts used to separate equal allocation, arm-specific geometric coverage, and common-skeleton coupling.

\subsection{Data-Generating Processes}
We design three observational data-generating processes (DGPs) of increasing complexity to assess the finite-sample behaviour of UD-DML across varying covariate geometries, degrees of treatment-effect heterogeneity, and treated-–control overlap regimes. Table \ref{tab:dgp_summary} summarizes the key mathematical structures. Across all configurations, the full data size is denoted by $n$, and the covariate dimension is fixed at $p=10$. The observed outcome $Y$ is generated as $Y = g(\X) + W \cdot \Delta(\X) + \varepsilon$ with $\varepsilon \sim \mathcal{N}(0, 1)$, and the true Average Treatment Effect $\theta_0 = \E[\Delta(\X)]$ is analytically fixed at $1.0$ in all three scenarios.

The first scenario (OBS-1) serves as a baseline of low heterogeneity and high overlap. The covariates are mutually independent and uniformly distributed, $X^{(d)}\overset{\text{i.i.d.}}{\sim}U[-2,2]$ for $d=1,\ldots,p$. The outcome mechanism is linear, with $g(\bm{X})=0.5X^{(1)}+0.3X^{(2)}$ and conditional average treatment effect (CATE) $\Delta(\bm{X})=1.0+0.2X^{(3)}$. Treatment is confounded by the logistic propensity $\operatorname{logit}(e(\bm{X}))=0.2X^{(1)}-0.2X^{(2)}$; the small coefficients keep propensity scores close to $0.5$, thereby ensuring a high degree of overlap between the treated and control groups.

The second scenario (OBS-2) introduces moderate heterogeneity alongside non-uniform covariate densities. The first five covariates follow $X^{(1)},\dots,X^{(5)} \iid U[-2, 2]$, while the remaining five follow $X^{(6)},\dots,X^{(10)} \iid \mathcal{N}(0, 1.5^2)$. The outcome baseline is moderately non-linear, $g(\X) = 0.5 (X^{(1)})^2 + 0.5 X^{(2)} X^{(3)} + \sin(X^{(6)})$, with an interactive CATE $\Delta(\X) = 1.0 + 0.5 X^{(1)} X^{(2)}$. The propensity score model, $\text{logit}(e(\X)) = 0.5 X^{(1)} - 0.3 (X^{(2)})^2 + 0.4 \sin(X^{(6)}) + 0.2 X^{(7)}$, combines linear, quadratic, and trigonometric terms, producing a more restricted yet manageable moderate-overlap region.

The third scenario (OBS-3) is designed to stress-test the estimators by combining strong heterogeneity with severe overlap deficiency. The first five covariates are drawn from an equally weighted two-component Gaussian mixture with component means $(-2,-2,0,0,0)$ and $(2,2,0,0,0)$ and per-coordinate standard deviation $0.5$, producing two sharply separated clusters; the remaining five covariates satisfy $X^{(6)},\ldots,X^{(10)}\overset{\text{i.i.d.}}{\sim}\mathcal{N}(0,1)$. The response surfaces are highly nonlinear, with $g(\bm{X})=\sin(\pi X^{(1)})+0.5X^{(2)}X^{(3)}+0.1(X^{(6)})^3+0.2\cos(X^{(7)})$ and $\Delta(\bm{X})=1.0+0.5\tanh(X^{(1)})+0.2X^{(6)}X^{(7)}$. The propensity score $\operatorname{logit}(e(\bm{X}))=0.3X^{(1)}+0.3X^{(2)}-0.5X^{(6)}$ acts on the bimodal mixture variables, systematically driving treatment probabilities towards the boundaries $0$ and $1$ and thus inducing a highly challenging low-overlap environment in which naive subsampling is expected to yield unreliable inference.

To isolate the effect of overlap severity on estimator performance, we additionally consider a tunable variant of OBS-3, in which the propensity score is rescaled as $\operatorname{logit}(e(\bm{X}))=c\cdot(0.3X^{(1)}+0.3X^{(2)}-0.5X^{(6)})$ for a multiplier $c\ge 0$. Setting $c=0$ yields perfect overlap ($e(\bm{X})\equiv 0.5$), $c=1$ recovers the default OBS-3, and larger values drive propensity scores towards $0$ and $1$. This DGP underlies the overlap-gradient experiment, enabling us to trace how the UD-DML advantage evolves as overlap deteriorates.

\subsection{Experimental Settings}\label{sec:simulation_settings}

Unless stated otherwise, the DML nuisance functions $\widehat{m}_g$ and $\widehat{e}$ are estimated by LightGBM \citep{Ke2017} with $100$ boosting rounds, maximum depth $5$, learning rate $0.1$, and $31$ leaves per tree. Cross-fitting uses $K=2$ folds, and estimated propensity scores are clipped to $[0.01,0.99]$ as a finite-sample stabilisation rule. For UD-DML, we retain principal components up to $\rho_0=0.85$ and search over $B_\gamma=30$ candidate generators. These quantities are fixed before outcomes are analysed; no method-specific outcome tuning is used. Each LightGBM fit uses \texttt{n\_jobs=1}; parallel execution is only across independent replications. Submission-scale summaries use $B=500$ Monte Carlo replications on a 12-vCPU Linux server with six parallel workers. A deterministic hierarchical seed identifies the DGP and replication, so the four ablation methods are evaluated on the same generated population in every replication. Exceptions are retained as structured failure records and reported rather than silently discarded.

The ablation study compares four working-sample designs under the same total budget $r$ and the same downstream DML specification. Every method samples without replacement. UNIF-DML samples $r$ observations without controlling arm size or covariate geometry. STRAT-UNIF-DML samples $r/2$ observations independently within each arm and therefore isolates the effect of equal allocation. SEP-UD-DML constructs a separate UD skeleton and empirical marginal transform within each arm, adding arm-specific geometric coverage without a common reference. UD-DML uses the pooled transform and one shared skeleton for both arms, thereby adding anchor-by-anchor coupling. This ordered sequence distinguishes the contributions of equal allocation, within-arm coverage, and common-skeleton alignment; the three adjacent contrasts are evaluated as paired changes because each method receives the same simulated population.

Finite-sample performance is evaluated by RMSE relative to the known simulated ATE, empirical coverage with Wilson intervals, mean interval width, the ratio of Monte Carlo to model-based standard errors, runtime, and failures. Every generated record stores the variance-method identifier. Submission claims about interval performance are permitted only after the selection-aware residual-variance implementation has been run at the prespecified scale; no legacy interval column is retained as evidence.

Covariate balance is summarised by the mean and maximum absolute treated--control standardised mean difference. For coordinate $d$, the reported quantity is
\[
\operatorname{SMD}_d
=
\frac{|\bar X_{1d}-\bar X_{0d}|}
{\{(s_{1d}^2+s_{0d}^2)/2\}^{1/2}}.
\]
For UD-DML, the realised GEFD between the common skeleton and the transformed pooled empirical population is approximated with $20{,}000$ independently drawn kernel pairs; its Monte Carlo standard error is computed on the squared-discrepancy scale. SEP-UD-DML instead constructs separate PCA and empirical-CDF systems within the two treatment arms, computes each skeleton's GEFD relative to its own arm-specific empirical distribution, and reports their arithmetic mean. The UD-DML and SEP-UD-DML values therefore use different coordinates and different reference measures. They are useful within-method representation diagnostics but are not an ordinal comparison between the two selectors. Both uniform-design methods report the mean and maximum anchor-matching radii. A separate single-process profile decomposes wall-clock time into standardisation and PCA, empirical-CDF sorting, skeleton search, inverse-CDF mapping, tree construction, matching, nuisance fitting, and inference; it also records peak resident-set size and the increase over the pre-call baseline.

The adjacent-stage evaluation uses the common Monte Carlo replication identifier, not pairs of selected observations. For method $m$ and replication $b$, define
\[
E_{m,b}=\left(\widehat\theta_{m,b}-\theta_0\right)^2.
\]
For a transition from method $A$ to method $B$, we report
\[
\widehat\Delta_{A\to B}
=\frac{1}{B}\sum_{b=1}^{B}\left(E_{B,b}-E_{A,b}\right),
\qquad
\widehat{\operatorname{MCSE}}
=\frac{\operatorname{sd}\{E_{B,b}-E_{A,b}:1\leq b\leq B\}}{\sqrt B}.
\]
Both methods receive the same generated population and hierarchical seed within replication $b$, which removes between-replication DGP variation from the contrast. A negative value indicates a lower mean squared estimation error for method $B$. Table~\ref{tab:simulation_summary} multiplies both quantities by $10^4$ for readability; this rescaling changes neither their sign nor their interpretation.

\subsection{Simulation Experiments}
The study comprises four experiments, each retaining the complete ordered comparison. They vary the working-sample size, the full-data size, nuisance specification, and overlap severity while holding the downstream DML procedure fixed. Collectively, they investigate (i) statistical efficiency and inferential validity as a function of $r$; (ii) scalability as $n$ grows while $r$ is held fixed; (iii) finite-sample sensitivity to nuisance-model misspecification; and (iv) sensitivity to treated--control overlap. The OBS-2 and OBS-3 canonical-budget cells also serve as the focused component-wise ablation.

The first experiment examines the behaviour of the estimators as a function of subsample size $r$. We fix $n=5\times 10^{5}$ and vary the total subsample size over $r\in\{1000,\,2500,\,5000,\,7500,\,10000\}$, so that $r_p=r/2$ treated--control pairs are selected at each level. For each of the three observational scenarios (OBS-1, OBS-2, OBS-3), we report RMSE, CI coverage, and CI width as functions of $r$ for all four methods.

The second experiment evaluates scalability while $r$ remains fixed and $n$ grows. We consider $r\in\{1000,5000\}$ and $n\in\{10^{5},5\times10^{5}\}$. Alongside all four working-sample methods, we report FULL-DML as an information-rich benchmark. The table displays RMSE, empirical coverage and interval width, and runtime, thereby separating the accuracy loss due to subsampling from the per-call computational savings.

The third experiment examines finite-sample sensitivity to nuisance misspecification; it is not a proof that double robustness survives data-dependent selection. We fix $n=5\times10^{5}$ and $r=5000$ and consider four learner configurations. ``Correct'' denotes the flexible LightGBM specification used in the main experiment. A restricted propensity learner is a logistic regression using only $X^{(5)}$, and a restricted outcome learner is ordinary linear regression on $X^{(5)}$ alone. The four configurations combine the flexible and restricted learners to show how the two subsampling designs behave when one or both nuisance learners omit important structure.

The fourth experiment assesses how the relative behaviour of all four working-sample designs changes with treated--control overlap. We use the OBS-3-overlap DGP with propensity-logit multiplier $c\in\{0.1,0.3,0.5,0.7,1.0,1.5\}$ at $n=5\times10^5$ and $r=5000$, and report RMSE, coverage, interval width, runtime, and overlap diagnostics without presupposing their ordering.

The ablation is conducted in OBS-2 and OBS-3 at $n=5\times10^5$ and $r=5000$, where non-uniform geometry and limited overlap make the distinctions among the sampling rules consequential. Contrasting UNIF-DML with STRAT-UNIF-DML measures the contribution of equal allocation; contrasting STRAT-UNIF-DML with SEP-UD-DML measures the contribution of arm-specific geometric coverage; and contrasting SEP-UD-DML with UD-DML measures the additional contribution of the pooled transform and common skeleton. The primary table reports RMSE, empirical coverage, mean interval width, mean absolute SMD, runtime, and failure counts. A companion diagnostic reports Wilson coverage intervals, empirical-to-model standard-error ratios, GEFD, matching radii, paired squared-error changes, component time, and peak memory. The order and metrics were fixed before the final simulation run.

\subsection{Results and Analysis} \label{sec:simulation_analysis}

The full-scale simulation suite completed in 64.97 hours using six parallel workers and LightGBM 4.6.0; the workflow manifest, including tests and post-processing, records 65.17 hours. Every scenario--method cell contained 500 successful Monte Carlo replications, the diagnostic ablation recorded no failures in 4000 attempts, and all 14 automated tests passed. Figure~\ref{fig:subsample_performance} reports the working-sample-size paths, Figure~\ref{fig:overlap_performance} reports the overlap-gradient experiment, and Table~\ref{tab:simulation_summary} gives the canonical $n=500{,}000$, $r=5000$ ablation for OBS-2 and OBS-3.

In OBS-2, equal allocation lowers RMSE from $0.0520$ for UNIF-DML to $0.0444$ for STRAT-UNIF-DML; separate arm-specific UD coverage gives $0.0437$, and common-skeleton coupling lowers it further to $0.0383$. On the $10^4\widehat\Delta$ scale used in Table~\ref{tab:simulation_summary}, the three paired changes are $-7.27$, $-0.58$, and $-4.47$, with scaled Monte Carlo standard errors $2.11$, $1.79$, and $1.56$. Thus the equal-allocation and common-skeleton steps drive the detectable improvement in this configuration, whereas the separate arm-specific geometry adds little beyond stratification.

OBS-3 provides the sharper ablation. STRAT-UNIF-DML only modestly improves on UNIF-DML, with RMSE decreasing from $0.0644$ to $0.0603$, and SEP-UD-DML increases RMSE to $0.0637$. UD-DML reduces RMSE to $0.0343$. On the same scaled squared-error scale, the SEP-UD-to-UD contrast is $-28.81$ with Monte Carlo standard error $2.93$, whereas the preceding STRAT-UNIF-to-SEP-UD contrast is positive. The result is therefore not explained by arm-specific space filling alone; the pooled transform and shared anchors are the consequential component under the bimodal, limited-overlap geometry.

Empirical coverage ranges from $0.940$ to $0.974$ across the eight canonical cells. These finite-sample values are compatible with nominal coverage but do not verify the assumptions of Theorem~\ref{thm:ud_root_r}. The design diagnostics are more directly tied to the construction. Relative to SEP-UD-DML, UD-DML reduces mean SMD from $0.1093$ to $0.0521$ in OBS-2 and from $0.2784$ to $0.0449$ in OBS-3; maximum SMD falls from $0.5632$ to $0.1159$ and from $1.1906$ to $0.1045$, respectively. By contrast, the larger reported UD-DML GEFD in OBS-2 is not evidence of worse cross-arm balance: it measures discrepancy from the pooled empirical distribution, whereas the smaller SEP-UD-DML number averages two discrepancies from separate arm-specific distributions in different transformed coordinates. The two GEFD values have different targets and should not be ranked against each other. Figure~\ref{fig:simulation_diagnostics} shows the corresponding propensity-support patterns without treating them as a proof of balance or overlap.

The nuisance-misspecification experiment in Table~\ref{tab:misspecification_summary} is a finite-sample stress test rather than a new double-robustness theorem. When both nuisance learners are deliberately restricted, UD-DML has RMSEs of $0.034$, $0.081$, and $0.041$ across OBS-1--OBS-3, with empirical coverages of $0.96$, $0.70$, and $0.95$. The corresponding UNIF-DML RMSEs are $0.062$, $0.195$, and $0.191$, and its coverages are $0.67$, $0.02$, and $0.00$. STRAT-UNIF-DML and SEP-UD-DML also deteriorate sharply in OBS-2 and OBS-3. These results indicate that the common-skeleton design can limit residual treated--control discrepancy when nuisance fitting is poor, but the protection is configuration-specific and incomplete, as shown by the subnominal OBS-2 coverage.

Across the working-sample grid in Figure~\ref{fig:subsample_performance}, all methods improve as $r$ increases. Their differences are small in OBS-1, but UD-DML is uniformly most accurate over the reported OBS-3 grid: its RMSE decreases from $0.1400$ at $r=1000$ to $0.0226$ at $r=10{,}000$, compared with $0.4034$ to $0.0329$ for UNIF-DML. In Figure~\ref{fig:overlap_performance}, UD-DML RMSE remains between $0.0325$ and $0.0357$ for $c\in\{0.1,0.3,0.5,0.7,1.0,1.5\}$, while the other three methods reach $0.0920$--$0.1017$ at $c=1.5$. This is finite-grid evidence that the advantage concentrates where treated--control overlap is most difficult. The coverage panel uses pointwise Wilson intervals rather than connected lines because adjacent coverage estimates are independent Monte Carlo summaries. UD-DML coverage ranges from $0.962$ to $0.982$ over this grid, indicating mild finite-sample conservatism rather than undercoverage; the intervals show the sampling uncertainty associated with 500 replications and prevent small Monte Carlo fluctuations from being read as a systematic trend.

The computational results in Table~\ref{tab:computation_summary} make the statistical cost explicit. At the canonical OBS-3 configuration, median UD-DML runtime is $2.093$ seconds and median peak RSS is $463.2$ MB. Tree construction and matching account for a median $72.2\%$ of runtime, whereas the cached skeleton search is negligible after its first use. Increasing $B_\gamma$ from 10 to 60 leaves mean runtime near $2.11$ seconds and produces RMSE values between $0.0325$ and $0.0392$ over 50 replications per setting. At $n=500{,}000$ and $r=5000$, the population-size experiment gives runtimes of $2.08$ seconds for UD-DML and $13.68$ seconds for FULL-DML in OBS-3, with RMSEs of $0.034$ and $0.004$, respectively. These values describe the accuracy--cost frontier generated by the working-sample budget.

\section{Real Data Application}\label{sec:real_data}
We next illustrate the fixed-budget comparison on a large observational dataset. We apply UNIF-DML, STRAT-UNIF-DML, SEP-UD-DML, UD-DML, and FULL-DML to the 2021 US Natality records. The FULL-DML estimate anchors the comparison of repeated-fit stability, reference agreement, and per-call runtime. Under consistency, conditional exchangeability given the listed pretreatment covariates, and positivity, the reported contrast identifies the ATE.

\subsection{Dataset and Preprocessing}\label{sec:real_data_setup}
We apply UD-DML to the 2021 US Natality microdata file released by the National Center for Health Statistics \citep{NCHS2021Natality} and available from the CDC VitalStats portal. The file records every live birth registered in the 50 states and the District of Columbia during the 2021 calendar year ($N\approx 3.67$ million records), together with detailed maternal, paternal, and infant characteristics drawn from the standard birth certificate. Our causal question of interest is the average treatment effect of maternal smoking during pregnancy on infant birth weight. This relationship has been extensively investigated in perinatal epidemiology \citep{Almond2005,Walker2009MaternalSmokingBW,DaSilva2019Smoking}, and serves as a natural large-scale benchmark for the present study: the treated group is small and unevenly distributed in the covariate space, yielding a genuinely low-overlap regime.

Formally, for each birth record $i=1,\ldots,n$, we define the treatment, outcome, and covariates as follows. The treatment indicator records whether the mother smoked during pregnancy (CDC field \texttt{CIG\_REC}): $W_i=1$ if reported ``Yes'' and $W_i=0$ if reported ``No''. The outcome $Y_i\in\mathbb{R}_+$ is the infant's birth weight in grams (\texttt{DBWT}), rescaled to the unit interval via the affine transformation $Y_i^{(s)}=(Y_i-y_{\min})/(y_{\max}-y_{\min})$ for LightGBM learning and clearer visualisation. The covariate vector $\bm{X}_i\in\mathbb{R}^{10}$ contains: maternal age (\texttt{MAGER}), maternal education (\texttt{MEDUC}), prenatal-care month (\texttt{PRECARE}), number of prenatal visits (\texttt{PREVIS}), infant sex (\texttt{SEX}), marital status (\texttt{DMAR}), paternal age (\texttt{FAGECOMB}), gestational diabetes (\texttt{RF\_GDIAB}), gestational hypertension (\texttt{RF\_GHYPE}), and the number of prior preterm births (\texttt{PRIORTERM}). The causal estimand is the population-level ATE on the scaled outcome,
$\theta_0={E}\left[Y^{(s)}(1)-Y^{(s)}(0)\right],$
which we report on the original scale via the back-transformation $(y_{\max}-y_{\min})\theta_0$.

Starting from the $3{,}669{,}928$ raw records, we apply the following cleaning pipeline: (i) drop records whose smoking indicator is missing, recoded, or marked as unknown/not stated (CDC code \texttt{U}), retaining only \texttt{Y}/\texttt{N} records; (ii) drop records with missing or implausible birth weight (\texttt{DBWT} $\le 0$ or $\ge 9000$,g); (iii) coerce every numeric covariate to numeric type and replace the documented sentinel codes by missing values; (iv) map the binary categorical fields \texttt{SEX}, \texttt{RF\_GDIAB}, and \texttt{RF\_GHYPE} to
$\{0,1\}$ and drop any remaining row containing a missing value; and (v) rescale \texttt{DBWT} to the unit interval as described above. The cleaned analysis sample contains $n=2{,}846{,}543$ births, of which $111{,}627$ ($3.92\%$) report maternal smoking; raw birth weights range over $[227,8142]$\,g. The target parameter is the scaled ATE $\theta_0$, and its back-transformation $(y_{\max}-y_{\min})\theta_0=7915\,\theta_0$ gives the identified mean birth-weight contrast in grams under these assumptions. A negative value corresponds to lower expected birth weight under maternal smoking than under no maternal smoking, conditional on identification.

\subsection{Experimental Settings}\label{sec:real_data_design}
The application is evaluated on one fixed cleaned dataset. We compute FULL-DML once to anchor the reference-agreement and runtime comparisons. We then perform $B=100$ paired repeated working-sample runs. At repetition $b$, UNIF-DML, STRAT-UNIF-DML, SEP-UD-DML, and UD-DML select and cross-fit working samples from the same fixed observed data for every $r\in\{1000,2500,5000,10000,25000\}$; the method--budget grid shares the repetition identifier and prespecified seed hierarchy. This design isolates variation due to working-sample selection and repeated fitting without mixing it with nonparametric-bootstrap variation.

For method $m$, let $\widehat\theta_{m,b}(r)$ denote the estimate in repetition $b$, and let $\widehat\theta_{\mathrm{FULL}}$ denote the single full-data reference. We report the repeated-fit mean and standard deviation together with
\[
\operatorname{RMSRef}_m(r)
=
\left\{
\frac{1}{B}
\sum_{b=1}^{B}
\left(
\widehat\theta_{m,b}(r)-\widehat\theta_{\mathrm{FULL}}
\right)^2
\right\}^{1/2}.
\]
This quantity measures repeated-fit agreement with the estimated FULL-DML benchmark. Bias relative to the causal effect and superpopulation uncertainty are separate inferential targets, so the application reports the benchmark discrepancy together with repeated-fit dispersion. It also records mean and maximum SMD, realised GEFD, matching radii, wall-clock time, and the variance method used in every fit. All nuisance fits use five-fold cross-fitting, the prespecified LightGBM configuration, one learner thread, and propensity clipping to $[0.01,0.99]$.

\subsection{Results and Analysis}\label{sec:real_data_results}
The full-data DML fit gives $\widehat\theta_{\mathrm{FULL}}=-0.01701$, with a 95\% interval of $[-0.01746,-0.01657]$. On the original outcome scale, the point estimate corresponds to a birth-weight contrast of approximately $-134.7$\,g. Under the identification conditions stated above, this estimate serves as the observational benchmark for the working-sample comparisons.

At the canonical budget $r=5000$, the repeated-fit means for UNIF-DML, STRAT-UNIF-DML, SEP-UD-DML, and UD-DML are $-0.01948$, $-0.01698$, $-0.01739$, and $-0.01601$, respectively. Their repeated-fit standard deviations are $0.01113$, $0.00403$, $0.00509$, and $0.00231$, while their RMS reference discrepancies are $0.01135$, $0.00401$, $0.00508$, and $0.00251$. Thus UD-DML is the most stable design and has the smallest RMS discrepancy from the estimated FULL-DML reference in this comparison. STRAT-UNIF-DML has the closest repeated-fit mean to the reference, but its larger dispersion produces a larger RMS discrepancy. These comparisons quantify stability and agreement with one estimated reference; they do not identify bias relative to the unknown causal effect.

Figure~\ref{fig:real_data_application} shows how this comparison changes with the working-sample budget. RMS reference discrepancy and repeated-fit dispersion decline for every method as $r$ increases. UD-DML remains lowest on both measures over the reported grid, reaching an RMS reference discrepancy of $0.00112$ and a repeated-fit standard deviation of $0.00092$ at $r=25000$. The canonical design diagnostics are consistent with the intended common-skeleton mechanism: UD-DML has mean SMD $0.0409$ and maximum SMD $0.1561$, compared with $0.1772$ and $0.8833$ for SEP-UD-DML. Its pooled realised GEFD is larger than the average arm-specific GEFD reported for SEP-UD-DML, but these discrepancies target different transformed empirical distributions and are not directly rankable.

The computational panel reports speed-up relative to the one-off FULL-DML runtime, with parity marked at one. The single FULL-DML fit takes $107.62$ seconds. At $r=5000$, mean runtimes are $0.68$, $0.74$, $32.23$, and $32.24$ seconds for the four working-sample designs in the same order, so UD-DML is about $3.3$ times faster at the canonical budget. At $r=10000$, its speed-up is about $1.6$. At $r=25000$, its mean runtime is $232.14$ seconds and its speed-up is $0.46$. The parity crossing identifies the empirical range of working-sample budgets that yields computational savings for repeated fitting.

\section{Conclusions}\label{sec:conclusion}
UD-DML provides a working-sample design for repeated DML under an explicit computational budget. The method builds a low-discrepancy skeleton in retained PCA coordinates and matches each anchor to one treated and one control observation, which combines equal allocation, geometric coverage, and arm alignment in a single selection rule. The theory establishes finite-sample representation and balance bounds, separates selected-sample estimation from transport to the original target, and derives a selection-conditional leading-score limit that yields root-$r$ asymptotic normality with a consistent residual variance estimator. The simulations attribute the largest gains in the more difficult designs to the common skeleton, and the natality analysis shows stable repeated-fit estimates together with useful computational savings over the relevant budget range.

The current analysis has several focused limitations. Dimension reduction is determined by an outcome-blind PCA threshold, so the retained coordinates do not need to capture every low-variance direction that is relevant to nuisance estimation. Nearest-neighbour matching also becomes more demanding as the retained dimension and working-sample budget increase. Finally, inferential theory states learner-agnostic calibration and componentwise rate conditions, verifying these rates for particular machine-learning procedures requires additional analysis. Empirical evidence is drawn from the simulation designs and one large observational application, which leaves the behavior under other covariate structures and treatment mechanisms to be assessed.

These limitations suggest concrete extensions. Dimension selection can combine retained-variance criteria with outcome-blind representation diagnostics, and faster or approximate matching algorithms can reduce the geometric cost in larger retained spaces. Learner-specific theory can translate the calibration conditions into primitive rates for commonly used nuisance estimators. Further work can also study alternative treatment settings and additional large observational datasets, using the same GEFD, matching-radius, and arm-balance diagnostics to determine how the common-skeleton construction transfers across applications.

\section{Declarations}
\paragraph{Availability of data and material}
The 2021 US Natality microdata used in this study is publicly available from the National Center for Health Statistics at the CDC VitalStats
portal \url{https://www.cdc.gov/nchs/data_access/vitalstatsonline.htm\#Tools}.
The codes used to implement the UD-DML estimator and reproduce the analysis in this manuscript can be accessed via \url{https://github.com/BobZhangHT/UD-DML}.

\paragraph{Funding}
This research is supported by program for scientific research startup funds (360302042403 and 360302042404) of Guangdong Ocean University, and also supported by Shenzhen Polytechnic University with the Project No.6025310026K.

\bibliographystyle{apalike}
\bibliography{references}

\clearpage

\appendix
\section{Technical Details}
\subsection{Proof of Lemma \ref{lem:ekh_rkhs}}\label{app:proof_ekh_rkhs}
\begin{proof}
For an indexed collection $\mathcal A=(\bm a_\ell)_{\ell=1}^{N_{\mathcal A}}$, define its kernel mean embedding by
\[
\mu_{\mathcal A}=\frac{1}{N_{\mathcal A}}
\sum_{\ell=1}^{N_{\mathcal A}}\mathcal K(\bm a_\ell,\,\cdot).
\]
The reproducing property gives
\[
\frac{1}{N_{\mathcal A}}\sum_{\ell=1}^{N_{\mathcal A}}f(\bm a_\ell)
-\frac{1}{N_{\mathcal B}}\sum_{s=1}^{N_{\mathcal B}}f(\bm b_s)
=\left\langle f,\mu_{\mathcal A}-\mu_{\mathcal B}\right\rangle_{\mathcal H_{\mathcal K}}.
\]
By Cauchy--Schwarz,
\[
\left|\left\langle f,\mu_{\mathcal A}-\mu_{\mathcal B}\right\rangle_{\mathcal H_{\mathcal K}}\right|
\leq
\|f\|_{\mathcal H_{\mathcal K}}
\|\mu_{\mathcal A}-\mu_{\mathcal B}\|_{\mathcal H_{\mathcal K}}.
\]
Expanding the squared norm gives
\[
\begin{aligned}
\|\mu_{\mathcal A}-\mu_{\mathcal B}\|_{\mathcal H_{\mathcal K}}^2
&=\frac{1}{N_{\mathcal A}^2}\sum_{\ell=1}^{N_{\mathcal A}}\sum_{\ell'=1}^{N_{\mathcal A}}
\mathcal K(\bm a_\ell,\bm a_{\ell'})
-\frac{2}{N_{\mathcal A}N_{\mathcal B}}
\sum_{\ell=1}^{N_{\mathcal A}}\sum_{s=1}^{N_{\mathcal B}}
\mathcal K(\bm a_\ell,\bm b_s)\\
&\quad+\frac{1}{N_{\mathcal B}^2}\sum_{s=1}^{N_{\mathcal B}}\sum_{s'=1}^{N_{\mathcal B}}
\mathcal K(\bm b_s,\bm b_{s'})
=D^2(\mathcal A;\mathcal B,\mathcal K).
\end{aligned}
\]
Therefore $\|\mu_{\mathcal A}-\mu_{\mathcal B}\|_{\mathcal H_{\mathcal K}}=D(\mathcal A;\mathcal B,\mathcal K)$, and the asserted inequality follows from $V_{\mathcal K}(f)=\|f\|_{\mathcal H_{\mathcal K}}$.
\end{proof}

\subsection{Proof of Proposition \ref{prop:copula_decomp}}\label{app:proof_copula_decomp}
\begin{proof}
Let $\mathcal H$ be the reproducing-kernel Hilbert space of $\mathcal K$. For an indexed collection $\mathcal A=(\bm a_\ell)_{\ell=1}^{N_{\mathcal A}}\subset[0,1]^q$, define
\[
\mu_{\mathcal A}=\frac{1}{N_{\mathcal A}}\sum_{\ell=1}^{N_{\mathcal A}}\mathcal K(\bm a_\ell,\,\cdot),
\qquad
\mu_u=\int_{[0,1]^q}\mathcal K(\bm u,\,\cdot)\,dF_u(\bm u).
\]
The relevant point-set embeddings are
\begin{align*}
\mu_{\widetilde{\mathcal V}_{r_p}}
=\frac1{r_p}\sum_{j=1}^{r_p}\mathcal K\!\left(T_{\bm Z}(\bm v_j),\,\cdot\right),
\quad
\mu_{\mathcal U_{r_p}}
=\frac1{r_p}\sum_{j=1}^{r_p}\mathcal K(\bm u_j,\,\cdot),
\quad
\mu_{\widetilde{\mathcal Z}}
=\frac1n\sum_{i=1}^{n}\mathcal K\!\left(T_{\bm Z}(\bm Z_i),\,\cdot\right).
\end{align*}
The definition in~\eqref{eq:gefd_sec24} gives
\[
D(\widetilde{\mathcal V}_{r_p};\widetilde{\mathcal Z},\mathcal K)
=\left\|\mu_{\widetilde{\mathcal V}_{r_p}}-\mu_{\widetilde{\mathcal Z}}\right\|_{\mathcal H}.
\]
Since
\[
\mu_{\widetilde{\mathcal V}_{r_p}}-\mu_{\widetilde{\mathcal Z}}
=\left(\mu_{\widetilde{\mathcal V}_{r_p}}-\mu_u\right)
+\left(\mu_u-\mu_{\widetilde{\mathcal Z}}\right),
\]
the triangle inequality yields
\[
D(\widetilde{\mathcal V}_{r_p};\widetilde{\mathcal Z},\mathcal K)
\leq D(\widetilde{\mathcal V}_{r_p};F_u,\mathcal K)
+D(\widetilde{\mathcal Z};F_u,\mathcal K).
\]
Similarly,
\[
\mu_{\widetilde{\mathcal V}_{r_p}}-\mu_u
=\left(\mu_{\widetilde{\mathcal V}_{r_p}}-\mu_{\mathcal U_{r_p}}\right)
+\left(\mu_{\mathcal U_{r_p}}-\mu_u\right),
\]
and therefore
\[
D(\widetilde{\mathcal V}_{r_p};F_u,\mathcal K)
\leq D(\widetilde{\mathcal V}_{r_p};\mathcal U_{r_p},\mathcal K)
+D(\mathcal U_{r_p};F_u,\mathcal K).
\]
Combining the two inequalities proves~\eqref{eq:copula_decomp}.
\end{proof}

\subsection{Proof of Proposition \ref{prop:target_decomp}}\label{app:proof_target_decomp}
\begin{proof}
The exact decomposition is
\[
\widehat\theta_{\mathrm{UD}}-\theta_0
=\bigl(\widehat\theta_{\mathrm{UD}}-\theta_{Q_{r,n}}\bigr)
+\bigl(\theta_{Q_{r,n}}-\theta_0\bigr).
\]
For the empirical analogue, define
\[
P_n h=\frac1n\sum_{i=1}^n h(\bm X_i),
\qquad
P_{\mathcal S}h=\frac1r\sum_{i\in\mathcal S}h(\bm X_i),
\qquad
\tau_P=f_\tau\circ\mathcal T+b_\tau.
\]
Then
\[
(P_{\mathcal S}-P_n)\tau_P
=(P_{\mathcal S}-P_n)(f_\tau\circ\mathcal T)
+(P_{\mathcal S}-P_n)b_\tau,
\]
and therefore
\begin{align*}
\left|(P_{\mathcal S}-P_n)\tau_P\right|
&\leq
\left|(P_{\mathcal S}-P_n)(f_\tau\circ\mathcal T)\right|
+\left|(P_{\mathcal S}-P_n)b_\tau\right|\\
&\leq
D(\widetilde{\mathcal V}_{r_p};\widetilde{\mathcal Z},\mathcal K)
V_{\mathcal K}(f_\tau)
+\frac{L_\tau}{2}
\left\{\delta_1^{(\mathrm{rot})}+\delta_0^{(\mathrm{rot})}\right\}
+\left|(P_{\mathcal S}-P_n)b_\tau\right|,
\end{align*}
where the last inequality is \eqref{eq:thm1_total} with $f=f_\tau$.
\end{proof}

\subsection{Derivation of the Closed Form of the Mixture Discrepancy}\label{sec:MD_deriv}

Let
\[
\mathcal D=(\bm u_j)_{j=1}^{r_p}\subset [0,1]^q,
\qquad
\bm u_j=(u_{j1},\ldots,u_{jq})^\top,
\]
and let
\[
F_{\mathcal D}
=
\frac{1}{r_p}\sum_{j=1}^{r_p}\delta_{\bm u_j}
\]
be the empirical distribution of the design points. Let $F_u$ denote the uniform distribution on $[0,1]^q$. For a reproducing kernel $\mathcal{K}:[0,1]^q\times[0,1]^q\to\mathbb R$, the squared discrepancy is
\begin{equation}\label{eq:app_disc_def}
D^2(\mathcal D;F_u,\mathcal{K})
=
\int_{[0,1]^{2q}}
\mathcal{K}(\bm u,\bm t)\,d\{F_{\mathcal D}(\bm u)-F_u(\bm u)\}\,d\{F_{\mathcal D}(\bm t)-F_u(\bm t)\}.
\end{equation}

We now derive the closed form of the mixture discrepancy $D_M^2(\mathcal D)$.

\begin{proposition}[Closed form of the mixture discrepancy]\label{prop:mixture_closed_form}
Let $\mathcal{K}_M$ be the mixture kernel
\[
\mathcal{K}_M(\bm u,\bm t)
=
\prod_{d=1}^q
\kappa_M(u_d,t_d),
\]
where
\[
\kappa_M(u,t)
=
\frac{15}{8}
-\frac{1}{4}\left|u-\frac{1}{2}\right|
-\frac{1}{4}\left|t-\frac{1}{2}\right|
-\frac{3}{4}|u-t|
+\frac{1}{2}|u-t|^2,
\qquad u,t\in[0,1].
\]
Then
\begin{align*}
D_M^2(\mathcal D)
&=
\left(\frac{19}{12}\right)^q
-\frac{2}{r_p}
\sum_{j=1}^{r_p}
\prod_{d=1}^q
\left[
\frac{5}{3}
-\frac{1}{4}\left|u_{jd}-\frac{1}{2}\right|
-\frac{1}{4}\left(u_{jd}-\frac{1}{2}\right)^2
\right] \\
&\quad
+\frac{1}{r_p^2}
\sum_{j=1}^{r_p}\sum_{k=1}^{r_p}
\prod_{d=1}^q
\left(
\frac{15}{8}
-\frac{1}{4}\left|u_{jd}-\frac{1}{2}\right|
-\frac{1}{4}\left|u_{kd}-\frac{1}{2}\right|
-\frac{3}{4}|u_{jd}-u_{kd}|
+\frac{1}{2}|u_{jd}-u_{kd}|^2
\right).
\end{align*}
\end{proposition}

\begin{proof}
Write
\[
\mu=F_{\mathcal D}-F_u,
\qquad
F_{\mathcal D}=\frac1{r_p}\sum_{j=1}^{r_p}\delta_{\bm u_j}.
\]
The signed-measure product expands as
\begin{align*}
d\mu(\bm u)d\mu(\bm t)
&=
\{dF_{\mathcal D}(\bm u)-dF_u(\bm u)\}
\{dF_{\mathcal D}(\bm t)-dF_u(\bm t)\}\\
&=
dF_{\mathcal D}(\bm u)dF_{\mathcal D}(\bm t)
-dF_{\mathcal D}(\bm u)dF_u(\bm t)\\
&\quad
-dF_u(\bm u)dF_{\mathcal D}(\bm t)
+dF_u(\bm u)dF_u(\bm t).
\end{align*}
Consequently,
\begin{align*}
D^2(\mathcal D;F_u,\mathcal K)
&=
\frac1{r_p^2}\sum_{j=1}^{r_p}\sum_{k=1}^{r_p}
\mathcal K(\bm u_j,\bm u_k)\\
&\quad
-\frac1{r_p}\sum_{j=1}^{r_p}
\int_{[0,1]^q}\mathcal K(\bm u_j,\bm t)\,d\bm t\\
&\quad
-\frac1{r_p}\sum_{k=1}^{r_p}
\int_{[0,1]^q}\mathcal K(\bm u,\bm u_k)\,d\bm u\\
&\quad
+\int_{[0,1]^{2q}}\mathcal K(\bm u,\bm t)\,d\bm u\,d\bm t.
\end{align*}
Since $\mathcal K(\bm u,\bm t)=\mathcal K(\bm t,\bm u)$,
\begin{equation}\label{eq:app_disc_analytic_short}
D^2(\mathcal D;F_u,\mathcal K)
=
\int_{[0,1]^{2q}}\mathcal K(\bm u,\bm t)\,d\bm u\,d\bm t
-\frac{2}{r_p}\sum_{j=1}^{r_p}\int_{[0,1]^q}\mathcal K(\bm u,\bm u_j)\,d\bm u
+\frac{1}{r_p^2}\sum_{j=1}^{r_p}\sum_{k=1}^{r_p}\mathcal K(\bm u_j,\bm u_k),
\end{equation}
Define
\[
A_0:=\int_0^1\int_0^1 \kappa_M(u,t)\,du\,dt,
\qquad
A_1(t):=\int_0^1 \kappa_M(u,t)\,du.
\]
For $t\in[0,1]$,
\begin{align*}
\int_0^1\left|u-\frac12\right|du
&=
\int_0^{1/2}\left(\frac12-u\right)du
+\int_{1/2}^{1}\left(u-\frac12\right)du\\
&=
\left[\frac u2-\frac{u^2}{2}\right]_{0}^{1/2}
+\left[\frac{u^2}{2}-\frac u2\right]_{1/2}^{1}\\
&=
\frac18+\frac18
=\frac14,
\end{align*}
\begin{align*}
\int_0^1|u-t|du
&=
\int_0^t(t-u)du+\int_t^1(u-t)du\\
&=
\left[tu-\frac{u^2}{2}\right]_0^t
+\left[\frac{u^2}{2}-tu\right]_t^1\\
&=
\frac{t^2}{2}+\frac{(1-t)^2}{2}\\
&=
t^2-t+\frac12\\
&=
\left(t-\frac12\right)^2+\frac14,
\end{align*}
and
\begin{align*}
\int_0^1|u-t|^2du
&=
\int_0^t(t-u)^2du+\int_t^1(u-t)^2du\\
&=
\left[-\frac{(t-u)^3}{3}\right]_0^t
+\left[\frac{(u-t)^3}{3}\right]_t^1\\
&=
\frac{t^3}{3}+\frac{(1-t)^3}{3}\\
&=
t^2-t+\frac13\\
&=
\left(t-\frac12\right)^2+\frac1{12}.
\end{align*}
Thus
\begin{align*}
A_1(t)
&=
\int_0^1\biggl\{
\frac{15}{8}
-\frac14\left|u-\frac12\right|
-\frac14\left|t-\frac12\right|
-\frac34|u-t|
+\frac12|u-t|^2
\biggr\}du\\
&=
\frac{15}{8}
-\frac14\left(\frac14\right)
-\frac14\left|t-\frac12\right|
-\frac34\left\{\left(t-\frac12\right)^2+\frac14\right\}\\
&\quad
+\frac12\left\{\left(t-\frac12\right)^2+\frac1{12}\right\}\\
&=
\left(\frac{15}{8}-\frac1{16}-\frac3{16}+\frac1{24}\right)
-\frac14\left|t-\frac12\right|
+\left(-\frac34+\frac12\right)\left(t-\frac12\right)^2\\
&=
\frac53
-\frac14\left|t-\frac12\right|
-\frac14\left(t-\frac12\right)^2.
\end{align*}
Moreover,
\begin{align*}
\int_0^1\left(t-\frac12\right)^2dt
&=
\left[\frac{(t-1/2)^3}{3}\right]_0^1
=\frac1{12},\\
A_0
&=
\int_0^1A_1(t)dt\\
&=
\frac53
-\frac14\int_0^1\left|t-\frac12\right|dt
-\frac14\int_0^1\left(t-\frac12\right)^2dt\\
&=
\frac53-\frac1{16}-\frac1{48}
=\frac{19}{12}.
\end{align*}
By Fubini's theorem and the product form of $\mathcal K_M$,
\[
\int_{[0,1]^{2q}}\mathcal K_M(\bm u,\bm t)\,d\bm u\,d\bm t
=\prod_{d=1}^q\int_0^1\!\int_0^1\kappa_M(u_d,t_d)\,du_d\,dt_d
=\prod_{d=1}^q A_0
=\left(\frac{19}{12}\right)^q.
\]
and
\begin{align*}
\int_{[0,1]^q}\mathcal K_M(\bm u,\bm u_j)\,d\bm u
&=
\prod_{d=1}^q\int_0^1\kappa_M(u_d,u_{jd})\,du_d\\
&=
\prod_{d=1}^q A_1(u_{jd})\\
&=
\prod_{d=1}^q
\left\{
\frac53-\frac14\left|u_{jd}-\frac12\right|
-\frac14\left(u_{jd}-\frac12\right)^2
\right\}.
\end{align*}
The empirical--empirical term is
\[
\mathcal K_M(\bm u_j,\bm u_k)
=
\prod_{d=1}^q
\left(
\frac{15}{8}
-\frac{1}{4}\Bigl|u_{jd}-\frac12\Bigr|
-\frac{1}{4}\Bigl|u_{kd}-\frac12\Bigr|
-\frac{3}{4}|u_{jd}-u_{kd}|
+\frac{1}{2}|u_{jd}-u_{kd}|^2
\right).
\]
Substitution into \eqref{eq:app_disc_analytic_short} gives
\begin{align*}
D_M^2(\mathcal D)
&=
\left(\frac{19}{12}\right)^q
-\frac{2}{r_p}\sum_{j=1}^{r_p}\prod_{d=1}^q
\left[
\frac53
-\frac14\left|u_{jd}-\frac12\right|
-\frac14\left(u_{jd}-\frac12\right)^2
\right] \\
&\quad
+\frac{1}{r_p^2}\sum_{j=1}^{r_p}\sum_{k=1}^{r_p}\prod_{d=1}^q
\left(
\frac{15}{8}
-\frac{1}{4}\left|u_{jd}-\frac12\right|
-\frac{1}{4}\left|u_{kd}-\frac12\right|
-\frac{3}{4}|u_{jd}-u_{kd}|
+\frac{1}{2}|u_{jd}-u_{kd}|^2
\right),
\end{align*}
which is the desired result.
\end{proof}

\subsection{Proof of Theorem \ref{thm:rep_ud}}
\begin{proof}
Fix $g\in\{0,1\}$. Since $\bm Z_i=\Pi_q(\bm X_i)$ for every $i$, the left-hand side of \eqref{eq:thm1_group} can be written as
\begin{align*}
\int \phi_f(\bm x)\,dP_{n,\bm X}(\bm x)
-
\int \phi_f(\bm x)\,dP_{\mathcal S_g,\bm X}(\bm x)
&=
\frac{1}{n}\sum_{i=1}^n
f\!\left(T_{\bm Z}(\bm Z_i)\right)
-
\frac{1}{r_p}\sum_{j=1}^{r_p}
f\!\left(T_{\bm Z}(\bm Z_{i_j^g})\right) \\
&=
\frac{1}{n}\sum_{i=1}^n \varphi_f(\bm Z_i)
-
\frac{1}{r_p}\sum_{j=1}^{r_p} \varphi_f(\bm Z_{i_j^g}).
\end{align*}
Add and subtract the skeleton average $r_p^{-1}\sum_{j=1}^{r_p}\varphi_f(\bm v_j)$ to obtain
\begin{align*}
&
\left|
\frac{1}{n}\sum_{i=1}^n \varphi_f(\bm Z_i)
-
\frac{1}{r_p}\sum_{j=1}^{r_p} \varphi_f(\bm Z_{i_j^g})
\right| \\
&\le
\left|
\frac{1}{n}\sum_{i=1}^n \varphi_f(\bm Z_i)
-
\frac{1}{r_p}\sum_{j=1}^{r_p} \varphi_f(\bm v_j)
\right|
+
\left|
\frac{1}{r_p}\sum_{j=1}^{r_p} \varphi_f(\bm v_j)
-
\frac{1}{r_p}\sum_{j=1}^{r_p} \varphi_f(\bm Z_{i_j^g})
\right| \\
&=: A_n+B_{n,g}.
\end{align*}

By Lemma~\ref{lem:ekh_rkhs}, applied to the transformed indexed collections $\widetilde{\mathcal Z}=(T_{\bm Z}(\bm Z_i))_{i=1}^n$ and $\widetilde{\mathcal V}_{r_p}=(T_{\bm Z}(\bm v_j))_{j=1}^{r_p}$,
\[
A_n
=
\left|
\frac{1}{n}\sum_{i=1}^n f\!\left(T_{\bm Z}(\bm Z_i)\right)
-
\frac{1}{r_p}\sum_{j=1}^{r_p} f\!\left(T_{\bm Z}(\bm v_j)\right)
\right|
\le
D(\widetilde{\mathcal V}_{r_p};\widetilde{\mathcal Z},\mathcal{K})\,V_{\mathcal K}(f).
\]

Next consider $B_{n,g}$. By the Lipschitz continuity of $\varphi_f$,
\begin{align*}
B_{n,g}
&\le
\frac{1}{r_p}\sum_{j=1}^{r_p}
\left|
\varphi_f(\bm v_j)-\varphi_f(\bm Z_{i_j^g})
\right| \\
&\le
\frac{L_f}{r_p}\sum_{j=1}^{r_p}
\|\bm v_j-\bm Z_{i_j^g}\|_2
\le
L_f\,\delta_g^{(\mathrm{rot})}.
\end{align*}
Thus
\begin{align*}
&\left|
\int\phi_f\,dP_{n,\bm X}
-\int\phi_f\,dP_{\mathcal S_g,\bm X}
\right|\\
&\qquad\leq
A_n+B_{n,g}\\
&\qquad\leq
D(\widetilde{\mathcal V}_{r_p};\widetilde{\mathcal Z},\mathcal K)V_{\mathcal K}(f)
+L_f\delta_g^{(\mathrm{rot})},
\end{align*}
which is \eqref{eq:thm1_group}.

It remains to prove \eqref{eq:thm1_total}. Since
\[
P_{\mathcal S,\bm X}
=
\frac{1}{2}P_{\mathcal S_1,\bm X}
+
\frac{1}{2}P_{\mathcal S_0,\bm X},
\]
we have
\begin{align*}
\int \phi_f(\bm x)\,dP_{n,\bm X}(\bm x)
-
\int \phi_f(\bm x)\,dP_{\mathcal S,\bm X}(\bm x) 
& = 
\frac{1}{2}
\left\{
\int \phi_f(\bm x)\,dP_{n,\bm X}(\bm x)
-
\int \phi_f(\bm x)\,dP_{\mathcal S_1,\bm X}(\bm x)
\right\} \\
\qquad &+
\frac{1}{2}
\left\{
\int \phi_f(\bm x)\,dP_{n,\bm X}(\bm x)
-
\int \phi_f(\bm x)\,dP_{\mathcal S_0,\bm X}(\bm x)
\right\}.
\end{align*}
Applying the triangle inequality and then \eqref{eq:thm1_group} with $g=1$ and $g=0$ gives
\[
\left|
\int \phi_f\,dP_{n,\bm X}
-
\int \phi_f\,dP_{\mathcal S,\bm X}
\right|
\le
D(\widetilde{\mathcal V}_{r_p};\widetilde{\mathcal Z},\mathcal{K})\,V_{\mathcal K}(f)
+
\frac{L_f}{2}
\bigl(
\delta_1^{(\mathrm{rot})}
+
\delta_0^{(\mathrm{rot})}
\bigr),
\]
which completes the proof.
\end{proof}

\subsection{Proof of Theorem \ref{thm:bal_ud}}
\begin{proof}
By construction, for each $g\in\{0,1\}$,
\[
\int \phi_f(\bm x)\,dP_{\mathcal S_g,\bm X}(\bm x)
=
\frac{1}{r_p}\sum_{j=1}^{r_p}\phi_f(\bm X_{i_j^g})
=
\frac{1}{r_p}\sum_{j=1}^{r_p}\varphi_f(\bm Z_{i_j^g}),
\]
where $\varphi_f(\bm z)=f\!\left(T_{\bm Z}(\bm z)\right)$. Therefore,
\begin{align*}
\left|
\int \phi_f(\bm x)\,dP_{\mathcal S_1,\bm X}(\bm x)
-
\int \phi_f(\bm x)\,dP_{\mathcal S_0,\bm X}(\bm x)
\right| 
=
\left|
\frac{1}{r_p}\sum_{j=1}^{r_p}\varphi_f(\bm Z_{i_j^1})
-
\frac{1}{r_p}\sum_{j=1}^{r_p}\varphi_f(\bm Z_{i_j^0})
\right|.
\end{align*}
Add and subtract the common skeleton average $r_p^{-1}\sum_{j=1}^{r_p}\varphi_f(\bm v_j)$ to obtain
\begin{align*}
&
\left|
\frac{1}{r_p}\sum_{j=1}^{r_p}\varphi_f(\bm Z_{i_j^1})
-
\frac{1}{r_p}\sum_{j=1}^{r_p}\varphi_f(\bm Z_{i_j^0})
\right| \\
&\le
\left|
\frac{1}{r_p}\sum_{j=1}^{r_p}\varphi_f(\bm Z_{i_j^1})
-
\frac{1}{r_p}\sum_{j=1}^{r_p}\varphi_f(\bm v_j)
\right|
+
\left|
\frac{1}{r_p}\sum_{j=1}^{r_p}\varphi_f(\bm v_j)
-
\frac{1}{r_p}\sum_{j=1}^{r_p}\varphi_f(\bm Z_{i_j^0})
\right|.
\end{align*}
Using the Lipschitz continuity of $\varphi_f$, the first term is bounded by
\[
\frac{1}{r_p}\sum_{j=1}^{r_p}
\left|
\varphi_f(\bm Z_{i_j^1})-\varphi_f(\bm v_j)
\right|
\le
\frac{L_f}{r_p}\sum_{j=1}^{r_p}\|\bm Z_{i_j^1}-\bm v_j\|_2
\le
L_f\,\delta_1^{(\mathrm{rot})},
\]
The second term satisfies
\begin{align*}
\left|
\frac{1}{r_p}\sum_{j=1}^{r_p}\varphi_f(\bm v_j)
-
\frac{1}{r_p}\sum_{j=1}^{r_p}\varphi_f(\bm Z_{i_j^0})
\right|
&\leq
\frac{1}{r_p}\sum_{j=1}^{r_p}
\left|\varphi_f(\bm v_j)-\varphi_f(\bm Z_{i_j^0})\right|\\
&\leq
\frac{L_f}{r_p}\sum_{j=1}^{r_p}
\|\bm v_j-\bm Z_{i_j^0}\|_2\\
&\leq
L_f\delta_0^{(\mathrm{rot})}.
\end{align*}
Therefore,
\begin{align*}
\left|
P_{\mathcal S_1,\bm X}\phi_f-P_{\mathcal S_0,\bm X}\phi_f
\right|
&\leq
L_f\delta_1^{(\mathrm{rot})}+L_f\delta_0^{(\mathrm{rot})}\\
&=
L_f\left\{\delta_1^{(\mathrm{rot})}+\delta_0^{(\mathrm{rot})}\right\},
\end{align*}
which is \eqref{eq:thm2_main}.
\end{proof}
\subsection{Proof of Corollary~\ref{cor:target_transport}}
\label{app:proof_target_transport}

\begin{proof}
Write
\[
P_{\mathcal S}
=
\frac12P_{\mathcal S_1}
+
\frac12P_{\mathcal S_0}.
\]
Because
\[
\tau_P=f_\tau\circ\mathcal T+b_\tau,
\]
we have
\begin{align*}
\left|
(P_{\mathcal S}-P_n)\tau_P
\right|
&\le
\left|
(P_{\mathcal S}-P_n)(f_\tau\circ\mathcal T)
\right|
+
\left|
(P_{\mathcal S}-P_n)b_\tau
\right|.
\end{align*}
Applying the pooled bound in
Theorem~\ref{thm:rep_ud} with $f=f_\tau$ gives
\begin{align*}
\left|
(P_{\mathcal S}-P_n)\tau_P
\right|
&\le
D(\widetilde{\mathcal V}_{r_p};
  \widetilde{\mathcal Z},\mathcal K)
V_{\mathcal K}(f_\tau)\\
&\quad+
\frac{L_\tau}{2}
\left\{
\delta_1^{(\mathrm{rot})}
+
\delta_0^{(\mathrm{rot})}
\right\}
+
\left|
(P_{\mathcal S}-P_n)b_\tau
\right|\\
&=
\Delta_{\tau,n}.
\end{align*}
Since
\[
\theta_{\mathcal S,n}=P_{\mathcal S}\tau_P,
\qquad
\theta_0=P\tau_P,
\]
it follows that
\[
\theta_{\mathcal S,n}-\theta_0
=
(P_{\mathcal S}-P_n)\tau_P
+
(P_n-P)\tau_P.
\]
Therefore,
\[
\sqrt r\,
\left|
\theta_{\mathcal S,n}-\theta_0
\right|
\le
\sqrt r\,\Delta_{\tau,n}
+
\sqrt r\,
\left|
(P_n-P)\tau_P
\right|.
\]
The finite-second-moment condition implies
\[
\sqrt n(P_n-P)\tau_P=O_P(1).
\]
Hence
\[
\sqrt r(P_n-P)\tau_P
=
O_P\!\left(\sqrt{\frac rn}\right)
=o_P(1),
\]
because $r/n\to0$. Combining this result with
$\sqrt r\,\Delta_{\tau,n}=o_P(1)$ proves
\[
\sqrt r(\theta_{\mathcal S,n}-\theta_0)=o_P(1).
\]
\end{proof}


\subsection{Proof of Theorem \ref{thm:ud_root_r}}\label{app:proof_ud_root_r}
\begin{proof}
All conditional quantities below are given $\mathcal G_n$.  For $i\in\mathcal S$,
\begin{align*}
E(\varepsilon_{W_i,i}\mid\mathcal G_n)
&=
W_iE\{Y_i-m_{1,P}(\bm X_i)\mid\bm X_i,W_i=1\}\\
&\quad+
(1-W_i)E\{Y_i-m_{0,P}(\bm X_i)\mid\bm X_i,W_i=0\}\\
&=0,
\end{align*}
and
\begin{align*}
E(\xi_{i,n}\mid\mathcal G_n)
&=
\frac{W_i}{\widehat e_i}
E(\varepsilon_{1,i}\mid\mathcal G_n)
-
\frac{1-W_i}{1-\widehat e_i}
E(\varepsilon_{0,i}\mid\mathcal G_n)\\
&=0.
\end{align*}
The distinct selected indices and outcome-blind selection give
\[
\{\xi_{i,n}:i\in\mathcal S\}
\quad\text{conditionally independent given }\mathcal G_n.
\]
Moreover,
\begin{align*}
|\xi_{i,n}|^{2+\nu}
&=
W_i\frac{|\varepsilon_{1,i}|^{2+\nu}}{\widehat e_i^{2+\nu}}
+(1-W_i)\frac{|\varepsilon_{0,i}|^{2+\nu}}
{(1-\widehat e_i)^{2+\nu}}\\
&\leq
\epsilon^{-(2+\nu)}|\varepsilon_{W_i,i}|^{2+\nu}.
\end{align*}
For every $c>0$,
\begin{align*}
L_{r,n}(c)
&:=
\frac1r\sum_{i\in\mathcal S}
E\!\left[
\xi_{i,n}^2\mathds1\{|\xi_{i,n}|>c\sqrt r\}
\mid\mathcal G_n
\right]\\
&\leq
\frac1r\sum_{i\in\mathcal S}
E\!\left[
\xi_{i,n}^2
\frac{|\xi_{i,n}|^\nu}{c^\nu r^{\nu/2}}
\mid\mathcal G_n
\right]\\
&=
\frac{1}{c^\nu r^{\nu/2}}
\frac1r\sum_{i\in\mathcal S}
E(|\xi_{i,n}|^{2+\nu}\mid\mathcal G_n)\\
&\leq
\frac{\epsilon^{-(2+\nu)}}{c^\nu r^{\nu/2}}
\frac1r\sum_{i\in\mathcal S}
E(|\varepsilon_{W_i,i}|^{2+\nu}\mid\mathcal G_n)\\
&=
r^{-\nu/2}O_P(1)
=o_P(1).
\end{align*}
The conditional variance satisfies
\[
\operatorname{Var}\!\left(
\frac1{\sqrt r}\sum_{i\in\mathcal S}\xi_{i,n}
\,\middle|\,\mathcal G_n
\right)
=
\frac1r\sum_{i\in\mathcal S}
E(\xi_{i,n}^2\mid\mathcal G_n)
=\sigma_{r,n}^2
\pto\sigma_{\rm UD}^2\in(0,\infty).
\]
Let $\{n_\ell\}_{\ell\geq1}$ be an arbitrary subsequence.  Convergence in
probability and a diagonal subsequence argument give a further subsequence
$\{n_{\ell_k}\}_{k\geq1}$ such that
\[
\sigma_{r,n_{\ell_k}}^2\to\sigma_{\rm UD}^2,
\qquad
L_{r,n_{\ell_k}}(c)\to0
\quad\text{almost surely for every }c\in\mathbb Q_{>0}.
\]
For $0<c_1<c_2$, monotonicity gives
\[
L_{r,n}(c_2)\leq L_{r,n}(c_1),
\qquad
L_{r,n_{\ell_k}}(c)\to0
\quad\text{almost surely for every }c>0.
\]
For almost every realised conditioning sequence, the deterministic
Lindeberg--Feller theorem applied to
\[
X_{n,i}=r^{-1/2}\xi_{i,n},
\qquad i\in\mathcal S,
\]
gives, for every $t\in\mathbb R$,
\[
E\!\left[
\exp\!\left\{
\frac{it}{\sqrt r}\sum_{i\in\mathcal S}\xi_{i,n_{\ell_k}}
\right\}
\middle|\mathcal G_{n_{\ell_k}}
\right]
\to
\exp\!\left(-\frac{t^2\sigma_{\rm UD}^2}{2}\right)
\quad\text{almost surely}.
\]
Since every subsequence has a further subsequence with this almost-sure limit,
for every $t\in\mathbb R$,
\begin{equation}\label{eq:ud_conditional_clt_cf}
E\!\left[
\exp\!\left\{
\frac{it}{\sqrt r}\sum_{i\in\mathcal S}\xi_{i,n}
\right\}
\middle|\mathcal G_n
\right]
\pto
\exp\!\left(-\frac{t^2\sigma_{\rm UD}^2}{2}\right).
\end{equation}
Let
\[
\chi_{r,n}(t)
=
E\!\left[
\exp\!\left\{
\frac{it}{\sqrt r}\sum_{i\in\mathcal S}\xi_{i,n}
\right\}
\middle|\mathcal G_n
\right],
\qquad
\chi(t)=\exp\!\left(-\frac{t^2\sigma_{\rm UD}^2}{2}\right).
\]
Since
\[
|\chi_{r,n}(t)|\leq1,
\qquad
|\chi(t)|\leq1,
\qquad
\chi_{r,n}(t)\pto\chi(t),
\]
for every $\eta>0$,
\begin{align*}
E|\chi_{r,n}(t)-\chi(t)|
&=
E\left[|\chi_{r,n}(t)-\chi(t)|
\mathds1\{|\chi_{r,n}(t)-\chi(t)|\leq\eta\}\right]\\
&\quad+
E\left[|\chi_{r,n}(t)-\chi(t)|
\mathds1\{|\chi_{r,n}(t)-\chi(t)|>\eta\}\right]\\
&\leq
\eta+2P\{|\chi_{r,n}(t)-\chi(t)|>\eta\}.
\end{align*}
Hence
\[
\limsup_{n\to\infty}E|\chi_{r,n}(t)-\chi(t)|\leq\eta,
\qquad
\lim_{n\to\infty}E|\chi_{r,n}(t)-\chi(t)|=0.
\]
Therefore,
\begin{align*}
E\exp\!\left\{
\frac{it}{\sqrt r}\sum_{i\in\mathcal S}\xi_{i,n}
\right\}
&=
E\{\chi_{r,n}(t)\}\\
&\to
\chi(t)
=
\exp\!\left(-\frac{t^2\sigma_{\rm UD}^2}{2}\right).
\end{align*}
Hence
\begin{equation}\label{eq:ud_conditional_clt}
\frac1{\sqrt r}\sum_{i\in\mathcal S}\xi_{i,n}
\dto\mathcal N(0,\sigma_{\rm UD}^2).
\end{equation}

Write
\[
\widehat m_{g,i}=\widehat m_g^{(-k_i)}(\bm X_i)
=m_{g,P}(\bm X_i)+\delta_{g,i},
\qquad g\in\{0,1\}.
\]
Then
\begin{align*}
Y_i-\widehat m_{1,i}
&=
Y_i-m_{1,P}(\bm X_i)-\delta_{1,i}
=\varepsilon_{1,i}-\delta_{1,i},\\
Y_i-\widehat m_{0,i}
&=
Y_i-m_{0,P}(\bm X_i)-\delta_{0,i}
=\varepsilon_{0,i}-\delta_{0,i},\\
\widehat m_{1,i}-\widehat m_{0,i}
&=
\tau_P(\bm X_i)+\delta_{1,i}-\delta_{0,i}.
\end{align*}
Substitution into the score yields
\begin{align*}
\widehat\psi_i^\ast
&=
\widehat m_{1,i}-\widehat m_{0,i}
+\frac{W_i(Y_i-\widehat m_{1,i})}{\widehat e_i}
-\frac{(1-W_i)(Y_i-\widehat m_{0,i})}{1-\widehat e_i}\\
&=
\tau_P(\bm X_i)+\delta_{1,i}-\delta_{0,i}
+\frac{W_i(\varepsilon_{1,i}-\delta_{1,i})}{\widehat e_i}
-\frac{(1-W_i)(\varepsilon_{0,i}-\delta_{0,i})}{1-\widehat e_i}\\
&=
\tau_P(\bm X_i)
+\frac{W_i\varepsilon_{1,i}}{\widehat e_i}
-\frac{(1-W_i)\varepsilon_{0,i}}{1-\widehat e_i}\\
&\quad+
\delta_{1,i}\left(1-\frac{W_i}{\widehat e_i}\right)
+\delta_{0,i}\left\{\frac{1-W_i}{1-\widehat e_i}-1\right\}\\
&=
\tau_P(\bm X_i)+\xi_{i,n}
+\delta_{1,i}\left(1-\frac{W_i}{\widehat e_i}\right)
+\delta_{0,i}\left\{\frac{1-W_i}{1-\widehat e_i}-1\right\}.
\end{align*}
Therefore,
\begin{align*}
\sqrt r(\widehat\theta_{\rm UD}-\theta_0)
&=
\sqrt r(\widehat\theta_{\rm UD}-\theta_{\mathcal S,n})
+\sqrt r(\theta_{\mathcal S,n}-\theta_0)\\
&=
\frac1{\sqrt r}\sum_{i\in\mathcal S}\xi_{i,n}
+\sqrt r B_{r,n}
+\sqrt r(\theta_{\mathcal S,n}-\theta_0)\\
&=
\frac1{\sqrt r}\sum_{i\in\mathcal S}\xi_{i,n}+o_P(1)+o_P(1)\\
&=
\frac1{\sqrt r}\sum_{i\in\mathcal S}\xi_{i,n}+o_P(1)
\dto\mathcal N(0,\sigma_{\rm UD}^2).
\end{align*}

For variance consistency, set
\[
d_{i,n}=\widehat\xi_i-\xi_{i,n}.
\]
Then
\begin{align*}
d_{i,n}
&=
-\frac{W_i\delta_{1,i}}{\widehat e_i}
+\frac{(1-W_i)\delta_{0,i}}{1-\widehat e_i},\\
d_{i,n}^2
&=
\frac{W_i\delta_{1,i}^2}{\widehat e_i^2}
+\frac{(1-W_i)\delta_{0,i}^2}{(1-\widehat e_i)^2}\\
&\leq
\epsilon^{-2}(\delta_{1,i}^2+\delta_{0,i}^2),\\
\frac1r\sum_{i\in\mathcal S}d_{i,n}^2
&\leq
\epsilon^{-2}\frac1r\sum_{i\in\mathcal S}
(\delta_{1,i}^2+\delta_{0,i}^2)
=o_P(1).
\end{align*}
Let
\[
\kappa=\min\left\{\frac\nu2,1\right\}\in(0,1],
\qquad
Z_{i,n}=\xi_{i,n}^2,
\qquad
Z_{i,n}^{(r)}=Z_{i,n}\mathds1\{Z_{i,n}\leq r\}.
\]
Because $2(1+\kappa)\leq2+\nu$,
\begin{align*}
\frac1r\sum_{i\in\mathcal S}
E(Z_{i,n}^{1+\kappa}\mid\mathcal G_n)
&=
\frac1r\sum_{i\in\mathcal S}
E(|\xi_{i,n}|^{2+2\kappa}\mid\mathcal G_n)\\
&\leq
\left\{
\frac1r\sum_{i\in\mathcal S}
E(|\xi_{i,n}|^{2+\nu}\mid\mathcal G_n)
\right\}^{\frac{2+2\kappa}{2+\nu}}\\
&=O_P(1),
\end{align*}
where the inequality follows from the power-mean inequality applied to the
uniform empirical distribution on $\mathcal S$ and the conditional laws.
The tail term obeys
\begin{align*}
\frac1r\sum_{i\in\mathcal S}
E\{Z_{i,n}\mathds1(Z_{i,n}>r)\mid\mathcal G_n\}
&\leq
\frac1{r^{1+\kappa}}
\sum_{i\in\mathcal S}E(Z_{i,n}^{1+\kappa}\mid\mathcal G_n)\\
&=
r^{-\kappa}O_P(1)
=o_P(1).
\end{align*}
The truncated conditional variance satisfies
\begin{align*}
\operatorname{Var}\!\left(
\frac1r\sum_{i\in\mathcal S}Z_{i,n}^{(r)}
\middle|\mathcal G_n
\right)
&=
\frac1{r^2}\sum_{i\in\mathcal S}
\operatorname{Var}(Z_{i,n}^{(r)}\mid\mathcal G_n)\\
&\leq
\frac1{r^2}\sum_{i\in\mathcal S}
E\{Z_{i,n}^2\mathds1(Z_{i,n}\leq r)\mid\mathcal G_n\}\\
&\leq
\frac{r^{1-\kappa}}{r^2}\sum_{i\in\mathcal S}
E(Z_{i,n}^{1+\kappa}\mid\mathcal G_n)\\
&=
r^{-\kappa}O_P(1)
=o_P(1).
\end{align*}
For every $\eta>0$, conditional Chebyshev's inequality gives
\begin{align*}
&P\!\left(
\left|
\frac1r\sum_{i\in\mathcal S}
\left[Z_{i,n}^{(r)}-E\{Z_{i,n}^{(r)}\mid\mathcal G_n\}\right]
\right|>\eta
\middle|\mathcal G_n
\right)\\
&\qquad\leq
\eta^{-2}
\operatorname{Var}\!\left(
\frac1r\sum_{i\in\mathcal S}Z_{i,n}^{(r)}
\middle|\mathcal G_n
\right)\\
&\qquad=
\eta^{-2}r^{-\kappa}O_P(1)
=o_P(1).
\end{align*}
Conditional Markov's inequality gives
\begin{align*}
&P\!\left(
\frac1r\sum_{i\in\mathcal S}Z_{i,n}\mathds1(Z_{i,n}>r)>\eta
\middle|\mathcal G_n
\right)\\
&\qquad\leq
\frac1\eta\frac1r\sum_{i\in\mathcal S}
E\{Z_{i,n}\mathds1(Z_{i,n}>r)\mid\mathcal G_n\}\\
&\qquad=
\eta^{-1}r^{-\kappa}O_P(1)
=o_P(1).
\end{align*}
Since each conditional probability belongs to $[0,1]$, convergence in probability
to zero implies convergence of its expectation to zero.  Hence
\begin{align*}
&\frac1r\sum_{i\in\mathcal S}
\left[Z_{i,n}^{(r)}-E\{Z_{i,n}^{(r)}\mid\mathcal G_n\}\right]
=o_P(1),\\
&\frac1r\sum_{i\in\mathcal S}Z_{i,n}\mathds1(Z_{i,n}>r)
=o_P(1),\\
&\frac1r\sum_{i\in\mathcal S}
E\{Z_{i,n}\mathds1(Z_{i,n}>r)\mid\mathcal G_n\}
=o_P(1).
\end{align*}
Consequently,
\begin{align*}
\frac1r\sum_{i\in\mathcal S}
\{\xi_{i,n}^2-E(\xi_{i,n}^2\mid\mathcal G_n)\}
&=
\frac1r\sum_{i\in\mathcal S}
\left[Z_{i,n}^{(r)}-E\{Z_{i,n}^{(r)}\mid\mathcal G_n\}\right]\\
&\quad+
\frac1r\sum_{i\in\mathcal S}Z_{i,n}\mathds1(Z_{i,n}>r)\\
&\quad-
\frac1r\sum_{i\in\mathcal S}
E\{Z_{i,n}\mathds1(Z_{i,n}>r)\mid\mathcal G_n\}\\
&=o_P(1),\\
\frac1r\sum_{i\in\mathcal S}\xi_{i,n}^2
&=
\sigma_{r,n}^2+o_P(1)
=
\sigma_{\rm UD}^2+o_P(1).
\end{align*}
Finally,
\begin{align*}
\left|
\frac1r\sum_{i\in\mathcal S}\widehat\xi_i^2
-\frac1r\sum_{i\in\mathcal S}\xi_{i,n}^2
\right|
&=
\left|
\frac1r\sum_{i\in\mathcal S}
\{2\xi_{i,n}d_{i,n}+d_{i,n}^2\}
\right|\\
&\leq
2\left(\frac1r\sum_{i\in\mathcal S}\xi_{i,n}^2\right)^{1/2}
\left(\frac1r\sum_{i\in\mathcal S}d_{i,n}^2\right)^{1/2}
+\frac1r\sum_{i\in\mathcal S}d_{i,n}^2\\
&=
2O_P(1)o_P(1)+o_P(1)
=o_P(1).
\end{align*}
Thus
\begin{align*}
\widehat\sigma_{\rm UD}^2
&=
\frac1r\sum_{i\in\mathcal S}\widehat\xi_i^2\\
&=\sigma_{\rm UD}^2+o_P(1),\\
\widehat\sigma_{\rm UD}
&=\left(\widehat\sigma_{\rm UD}^2\right)^{1/2}
\pto\sigma_{\rm UD}>0,\\
\frac{\widehat\theta_{\rm UD}-\theta_0}{\sqrt{\widehat V_{\rm UD}}}
&=
\frac{\sqrt r(\widehat\theta_{\rm UD}-\theta_0)}
{\widehat\sigma_{\rm UD}}
\dto\mathcal N(0,1),\\
P\!\left(
\theta_0\in
\left[\widehat\theta_{\rm UD}
\pm z_{1-\alpha/2}\sqrt{\widehat V_{\rm UD}}\right]
\right)
&\to
\Phi(z_{1-\alpha/2})-\Phi(-z_{1-\alpha/2})\\
&=1-\alpha.
\end{align*}
\end{proof}

\clearpage

\begin{sidewaystable}[p]
\centering
\caption{Summary of the data-generating processes. In every scenario $p=10$, $\theta_0=1$, and $Y=g(\bm X)+W\Delta(\bm X)+\varepsilon$ with $\varepsilon\sim\mathcal N(0,1)$. Treatment follows the listed logistic propensity model.}
\label{tab:dgp_summary}
\resizebox{\textwidth}{!}{%
\renewcommand{\arraystretch}{1.5}
\begin{tabular}{@{}lllll@{}}
\toprule
Scenario & Covariates $\bm X$ & Baseline $g(\bm X)$ & CATE $\Delta(\bm X)$ & $\operatorname{logit}\{e(\bm X)\}$ \\
\midrule
OBS-1 & $X^{(d)}\iid U[-2,2]$ & $0.5X^{(1)}+0.3X^{(2)}$ & $1+0.2X^{(3)}$ & $0.2X^{(1)}-0.2X^{(2)}$ \\
OBS-2 & $X^{(1:5)}\iid U[-2,2]$; $X^{(6:10)}\iid\mathcal N(0,1.5^2)$ & $0.5(X^{(1)})^2+0.5X^{(2)}X^{(3)}+\sin X^{(6)}$ & $1+0.5X^{(1)}X^{(2)}$ & $0.5X^{(1)}-0.3(X^{(2)})^2+0.4\sin X^{(6)}+0.2X^{(7)}$ \\
OBS-3 & $X^{(1:5)}\sim\operatorname{GMM}(\bm\mu_1,\bm\mu_2;0.5)$; $X^{(6:10)}\iid\mathcal N(0,1)$ & $\sin(\pi X^{(1)})+0.5X^{(2)}X^{(3)}+0.1(X^{(6)})^3+0.2\cos X^{(7)}$ & $1+0.5\tanh X^{(1)}+0.2X^{(6)}X^{(7)}$ & $0.3X^{(1)}+0.3X^{(2)}-0.5X^{(6)}$ \\
\bottomrule
\multicolumn{5}{l}{\footnotesize For OBS-3, $\bm\mu_{1,2}=(\mp2,\mp2,0,0,0)$ and the two Gaussian components have equal weights.}
\end{tabular}}
\end{sidewaystable}

\begin{sidewaystable}[p]
\centering
\footnotesize
\caption{Full-scale canonical-budget ablation at $n=500{,}000$ and $r=5000$. Every scenario--method cell contains 500 successful Monte Carlo replications, with no failures among the 4000 attempted ablation fits. For method $m$ in replication $b$, let $E_{m,b}=(\widehat\theta_{m,b}-\theta_0)^2$. For a transition $A\to B$, the paired contrast is the mean of $E_{B,b}-E_{A,b}$ over common replication identifiers; a negative value indicates smaller mean squared estimation error for method $B$. The contrast and its MCSE are multiplied by $10^4$ in the table. The GEFD entries are target-specific: UD-DML uses the pooled transform and pooled empirical reference, whereas SEP-UD-DML averages two arm-specific GEFDs computed in separate transforms. These entries must not be ranked across the two methods.}
\label{tab:simulation_summary}
\setlength{\tabcolsep}{4.5pt}
\begin{tabular*}{0.98\linewidth}{@{\extracolsep{\fill}}lllrrrrr@{}}
\toprule
\multicolumn{8}{c}{\textit{Canonical-budget estimation and inference}} \\
\cmidrule(lr){1-8}
Scenario & Method & Replications & RMSE & Coverage & Mean CI width & Mean SMD & Mean runtime (s) \\
\midrule
OBS-2 & UNIF-DML       & 500 & 0.0520 & 0.952 & 0.2049 & 0.1140 & 0.320 \\
OBS-2 & STRAT-UNIF-DML & 500 & 0.0444 & 0.954 & 0.1852 & 0.1136 & 0.321 \\
OBS-2 & SEP-UD-DML     & 500 & 0.0437 & 0.960 & 0.1815 & 0.1093 & 2.711 \\
OBS-2 & UD-DML         & 500 & 0.0383 & 0.940 & 0.1500 & 0.0521 & 2.631 \\
OBS-3 & UNIF-DML       & 500 & 0.0644 & 0.956 & 0.2562 & 0.2879 & 0.310 \\
OBS-3 & STRAT-UNIF-DML & 500 & 0.0603 & 0.966 & 0.2544 & 0.2880 & 0.311 \\
OBS-3 & SEP-UD-DML     & 500 & 0.0637 & 0.958 & 0.2562 & 0.2784 & 2.182 \\
OBS-3 & UD-DML         & 500 & 0.0343 & 0.974 & 0.1440 & 0.0449 & 2.078 \\
\midrule
\multicolumn{8}{c}{\textit{Canonical-budget design diagnostics}} \\
\cmidrule(lr){1-8}
Scenario & Budget & Method & Mean SMD & Maximum SMD & Target-specific GEFD & Mean radius & Maximum radius \\
\midrule
OBS-2 & $r=5000$ & UNIF-DML       & 0.1140 & 0.5368 & --     & --     & --     \\
OBS-2 & $r=5000$ & STRAT-UNIF-DML & 0.1136 & 0.5392 & --     & --     & --     \\
OBS-2 & $r=5000$ & SEP-UD-DML     & 0.1093 & 0.5632 & 0.1412 & 0.9012 & 2.1939 \\
OBS-2 & $r=5000$ & UD-DML         & 0.0521 & 0.1159 & 0.2919 & 0.8993 & 2.1946 \\
OBS-3 & $r=5000$ & UNIF-DML       & 0.2879 & 1.1881 & --     & --     & --     \\
OBS-3 & $r=5000$ & STRAT-UNIF-DML & 0.2880 & 1.1886 & --     & --     & --     \\
OBS-3 & $r=5000$ & SEP-UD-DML     & 0.2784 & 1.1906 & 0.1830 & 0.6594 & 1.6044 \\
OBS-3 & $r=5000$ & UD-DML         & 0.0449 & 0.1045 & 0.1020 & 0.6680 & 1.7909 \\
\midrule
\multicolumn{8}{c}{\textit{Paired adjacent-stage contrasts}} \\
\cmidrule(lr){1-8}
Scenario & \multicolumn{2}{l}{Transition $A\to B$} & Mean $E_B-E_A$ ($\times10^4$) & MCSE ($\times10^4$) & Paired replications & \multicolumn{2}{c}{ } \\
\midrule
OBS-2 & \multicolumn{2}{l}{UNIF $\to$ STRAT-UNIF} & $-7.2720$ & $2.1089$ & 500 & \multicolumn{2}{c}{ } \\
OBS-2 & \multicolumn{2}{l}{STRAT-UNIF $\to$ SEP-UD} & $-0.5849$ & $1.7908$ & 500 & \multicolumn{2}{c}{ } \\
OBS-2 & \multicolumn{2}{l}{SEP-UD $\to$ UD} & $-4.4656$ & $1.5594$ & 500 & \multicolumn{2}{c}{ } \\
OBS-3 & \multicolumn{2}{l}{UNIF $\to$ STRAT-UNIF} & $-5.1628$ & $3.5406$ & 500 & \multicolumn{2}{c}{ } \\
OBS-3 & \multicolumn{2}{l}{STRAT-UNIF $\to$ SEP-UD} & $4.2643$ & $3.5667$ & 500 & \multicolumn{2}{c}{ } \\
OBS-3 & \multicolumn{2}{l}{SEP-UD $\to$ UD} & $-28.8110$ & $2.9324$ & 500 & \multicolumn{2}{c}{ } \\
\bottomrule
\end{tabular*}
\end{sidewaystable}

\begin{sidewaystable}[p]
\centering
\footnotesize
\caption{Full-scale nuisance-misspecification stress test at $n=500{,}000$ and $r=5000$, with 500 replications per cell. ``Flexible'' denotes the prespecified LightGBM learner and ``Restricted'' denotes the deliberately misspecified learner using only the restricted covariate input. The experiment is a finite-sample diagnostic and is not presented as a selection-aware double-robustness theorem.}
\label{tab:misspecification_summary}
\setlength{\tabcolsep}{4.5pt}
\begin{tabular*}{0.98\linewidth}{@{\extracolsep{\fill}}lllrrrrrrrr@{}}
\toprule
& & & \multicolumn{4}{c}{RMSE} & \multicolumn{4}{c}{Empirical coverage} \\
\cmidrule(lr){4-7}\cmidrule(lr){8-11}
Scenario & Outcome learner & Propensity learner & UNIF & STRAT-UNIF & SEP-UD & UD & UNIF & STRAT-UNIF & SEP-UD & UD \\
\midrule
OBS-1 & Flexible   & Flexible   & 0.033 & 0.034 & 0.035 & 0.033 & 0.97 & 0.97 & 0.95 & 0.96 \\
OBS-1 & Flexible   & Restricted & 0.029 & 0.029 & 0.029 & 0.030 & 0.96 & 0.97 & 0.96 & 0.95 \\
OBS-1 & Restricted & Flexible   & 0.037 & 0.037 & 0.037 & 0.033 & 0.98 & 0.98 & 0.98 & 0.98 \\
OBS-1 & Restricted & Restricted & 0.062 & 0.063 & 0.057 & 0.034 & 0.67 & 0.65 & 0.73 & 0.96 \\
\addlinespace
OBS-2 & Flexible   & Flexible   & 0.052 & 0.044 & 0.044 & 0.038 & 0.95 & 0.95 & 0.96 & 0.94 \\
OBS-2 & Flexible   & Restricted & 0.038 & 0.034 & 0.033 & 0.034 & 0.91 & 0.93 & 0.93 & 0.92 \\
OBS-2 & Restricted & Flexible   & 0.068 & 0.068 & 0.065 & 0.041 & 0.98 & 0.95 & 0.95 & 0.99 \\
OBS-2 & Restricted & Restricted & 0.195 & 0.192 & 0.179 & 0.081 & 0.02 & 0.02 & 0.00 & 0.70 \\
\addlinespace
OBS-3 & Flexible   & Flexible   & 0.064 & 0.060 & 0.064 & 0.034 & 0.96 & 0.97 & 0.96 & 0.97 \\
OBS-3 & Flexible   & Restricted & 0.043 & 0.042 & 0.045 & 0.029 & 0.85 & 0.84 & 0.81 & 0.97 \\
OBS-3 & Restricted & Flexible   & 0.133 & 0.133 & 0.151 & 0.036 & 0.78 & 0.76 & 0.69 & 0.99 \\
OBS-3 & Restricted & Restricted & 0.191 & 0.192 & 0.206 & 0.041 & 0.00 & 0.00 & 0.00 & 0.95 \\
\bottomrule
\end{tabular*}
\end{sidewaystable}

\begin{sidewaystable}[p]
\centering
\small
\caption{Full-scale UD-DML computational profile and generator-budget sensitivity in OBS-3. The profile uses $n=500{,}000$, $r=5000$, $B_\gamma=30$, 50 single-process replications, and one learner thread. The cached skeleton search has negligible median cost; its one-time generation cost is included in the raw profile. Memory is resident-set size (RSS). The generator-budget experiment uses 50 replications per setting.}
\label{tab:computation_summary}
\setlength{\tabcolsep}{4.5pt}
\begin{tabular*}{0.96\linewidth}{@{\extracolsep{\fill}}lllrrr@{}}
\toprule
\multicolumn{6}{c}{\textit{UD-DML runtime decomposition}} \\
\cmidrule(lr){1-6}
\multicolumn{3}{l}{Component} & Median time (s) & IQR (s) & Median share (\%) \\
\midrule
\multicolumn{3}{l}{Standardisation and PCA} & 0.0934 & 0.0108 & 4.40 \\
\multicolumn{3}{l}{Empirical-CDF sorting} & 0.0635 & 0.0016 & 3.05 \\
\multicolumn{3}{l}{Cached skeleton search} & 0.0000 & 0.0000 & 0.00 \\
\multicolumn{3}{l}{Inverse-CDF mapping} & 0.0004 & 0.0000 & 0.02 \\
\multicolumn{3}{l}{Tree construction and matching} & 1.5083 & 0.0260 & 72.16 \\
\multicolumn{3}{l}{Nuisance fitting} & 0.2976 & 0.0059 & 14.22 \\
\multicolumn{3}{l}{Inference} & 0.0008 & 0.0000 & 0.04 \\
\multicolumn{3}{l}{Other orchestration and measurement} & 0.1291 & 0.0061 & 6.16 \\
\multicolumn{3}{l}{Total} & 2.0919 & 0.0293 & 100.00 \\
\midrule
\multicolumn{6}{c}{\textit{Memory summary}} \\
\cmidrule(lr){1-6}
\multicolumn{3}{l}{Configuration} & Median peak RSS (MB) & Median incremental peak RSS (MB) & Replications \\
\midrule
\multicolumn{3}{l}{OBS-3, $n=500{,}000$, $r=5000$, $B_\gamma=30$} & 463.2 & 85.3 & 50 \\
\midrule
\multicolumn{6}{c}{\textit{Sensitivity to the generator budget}} \\
\cmidrule(lr){1-6}
$B_\gamma$ & Mean runtime (s) & Median runtime (s) & RMSE & Mean CI width & Coverage \\
\midrule
10 & 2.107 & 2.098 & 0.0331 & 0.1425 & 0.960 \\
20 & 2.129 & 2.066 & 0.0325 & 0.1435 & 0.960 \\
30 & 2.113 & 2.090 & 0.0392 & 0.1437 & 0.940 \\
40 & 2.120 & 2.081 & 0.0359 & 0.1445 & 0.940 \\
60 & 2.112 & 2.064 & 0.0370 & 0.1447 & 0.980 \\
\bottomrule
\end{tabular*}
\end{sidewaystable}

\clearpage

\begin{figure}[p]
\centering
\IfFileExists{figures/subsample_size_metrics.pdf}{\includegraphics[width=0.88\textwidth]{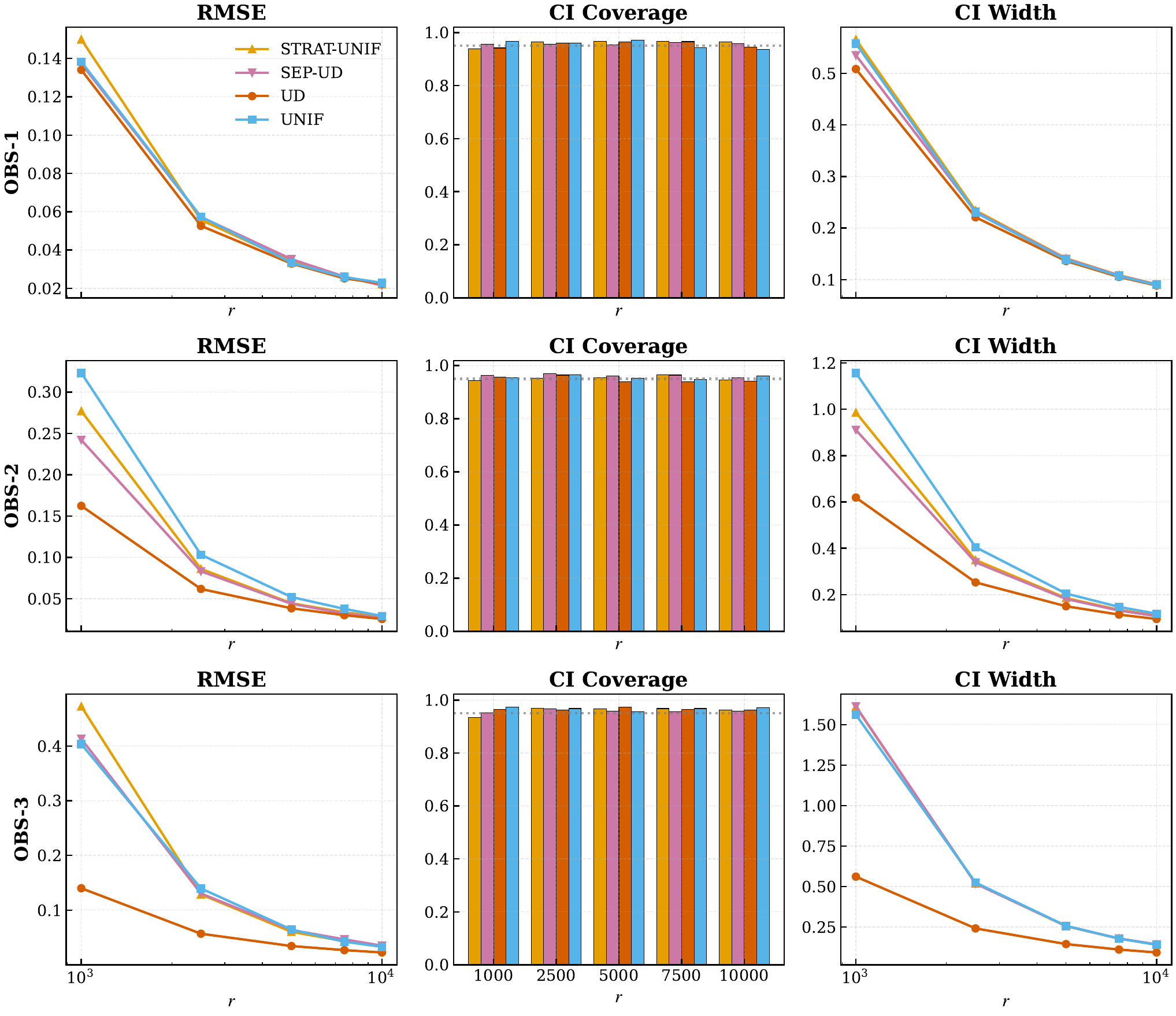}}{\ResultPlaceholder{Expected full-scale asset: \path{figures/subsample_size_metrics.pdf}}}
\caption{Full-scale working-sample-size experiment for all four designs. Rows correspond to OBS-1--OBS-3; columns report RMSE, empirical coverage, and mean interval width over $r\in\{1000,2500,5000,7500,10000\}$ at $n=500{,}000$. Every plotted cell contains 500 successful replications.}
\label{fig:subsample_performance}
\end{figure}

\begin{figure}[p]
\centering
\IfFileExists{figures/overlap_gradient.pdf}{\includegraphics[width=0.92\textwidth]{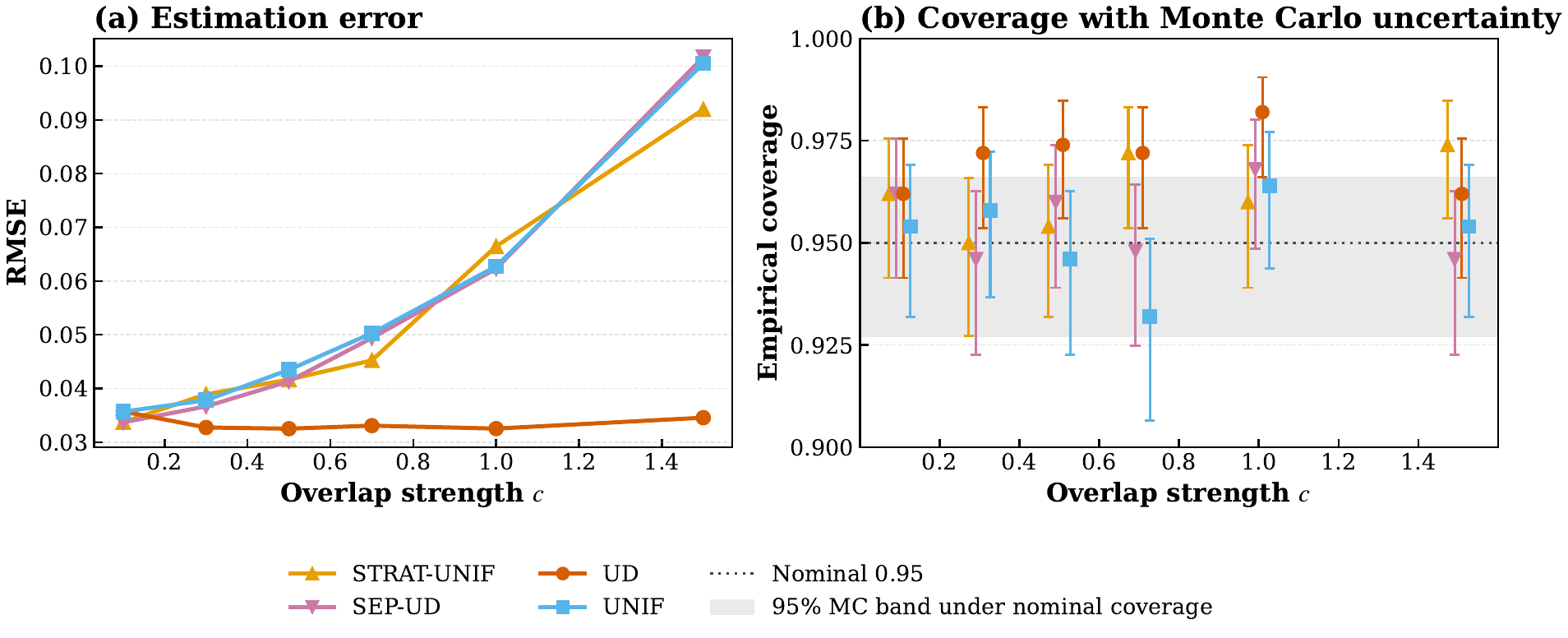}}{\ResultPlaceholder{Expected full-scale asset: \path{figures/overlap_gradient.pdf}}}
\caption{Full-scale overlap-gradient experiment for all four working-sample designs in OBS-3. The propensity-logit multiplier varies over $c\in\{0.1,0.3,0.5,0.7,1.0,1.5\}$ at $n=500{,}000$ and $r=5000$. Panel (a) reports RMSE. Panel (b) reports empirical coverage as horizontally offset points with pointwise 95\% Wilson intervals; points are not connected because adjacent values are independent Monte Carlo estimates. The dotted line marks nominal 0.95 coverage, and the grey band gives the 95\% Monte Carlo range under nominal coverage for 500 replications.}
\label{fig:overlap_performance}
\end{figure}

\begin{figure}[p]
\centering
\IfFileExists{figures/visualization_propensity_density.pdf}{\includegraphics[width=0.90\textwidth]{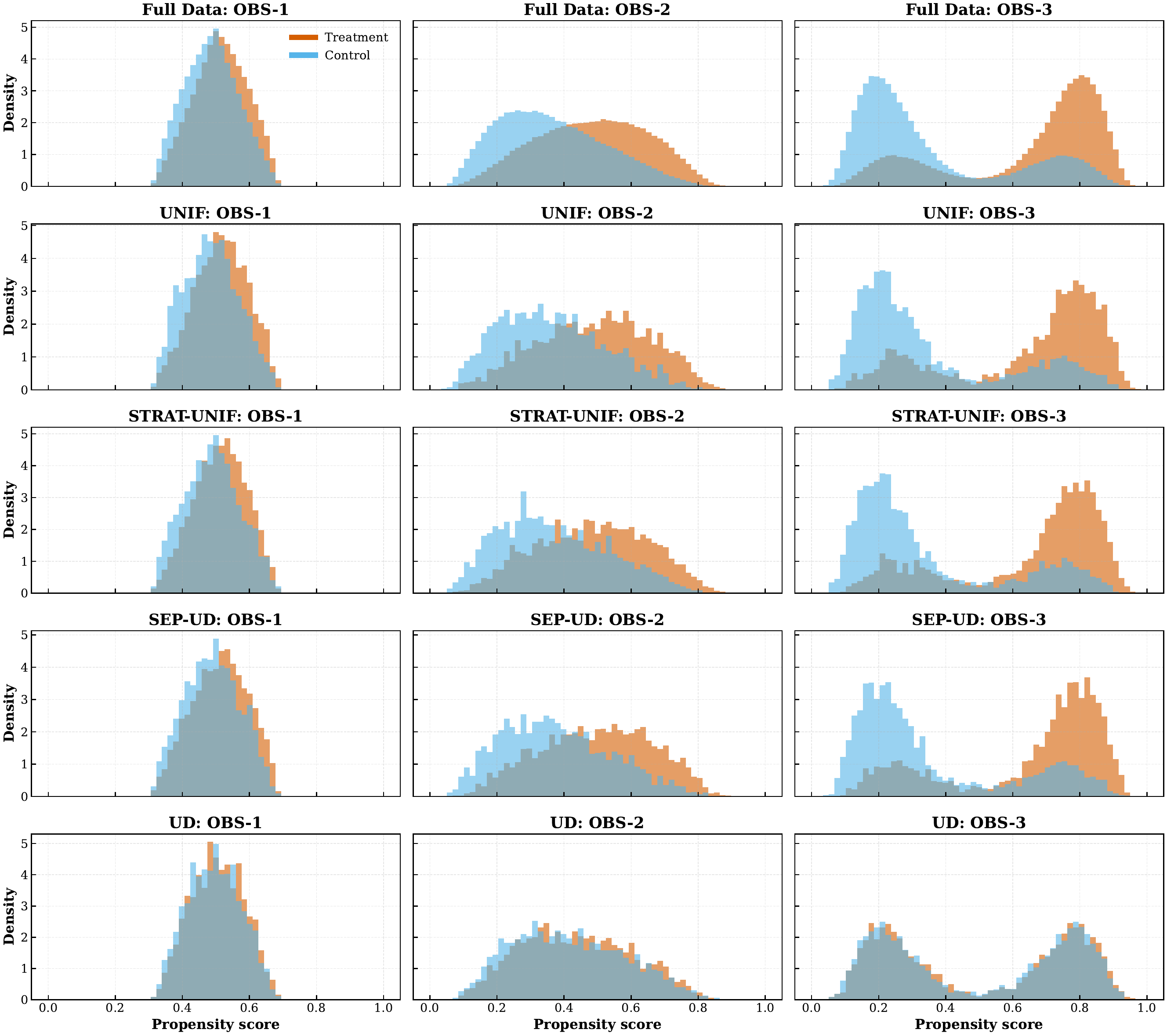}}{\ResultPlaceholder{Expected full-scale asset: \path{figures/visualization_propensity_density.pdf}}}
\caption{Full-scale propensity-score support diagnostic for the full data and the four working-sample designs in OBS-1--OBS-3 at $n=500{,}000$ and $r=5000$. The distributions show how each selection rule represents treated and control support. Numerical covariate-balance summaries, realised GEFD, and matching radii are reported in Table~\ref{tab:simulation_summary}; the histograms are descriptive and do not by themselves establish positivity or the assumptions of Theorem~\ref{thm:ud_root_r}.}
\label{fig:simulation_diagnostics}
\end{figure}

\begin{figure}[p]
\centering
\includegraphics[width=\textwidth]{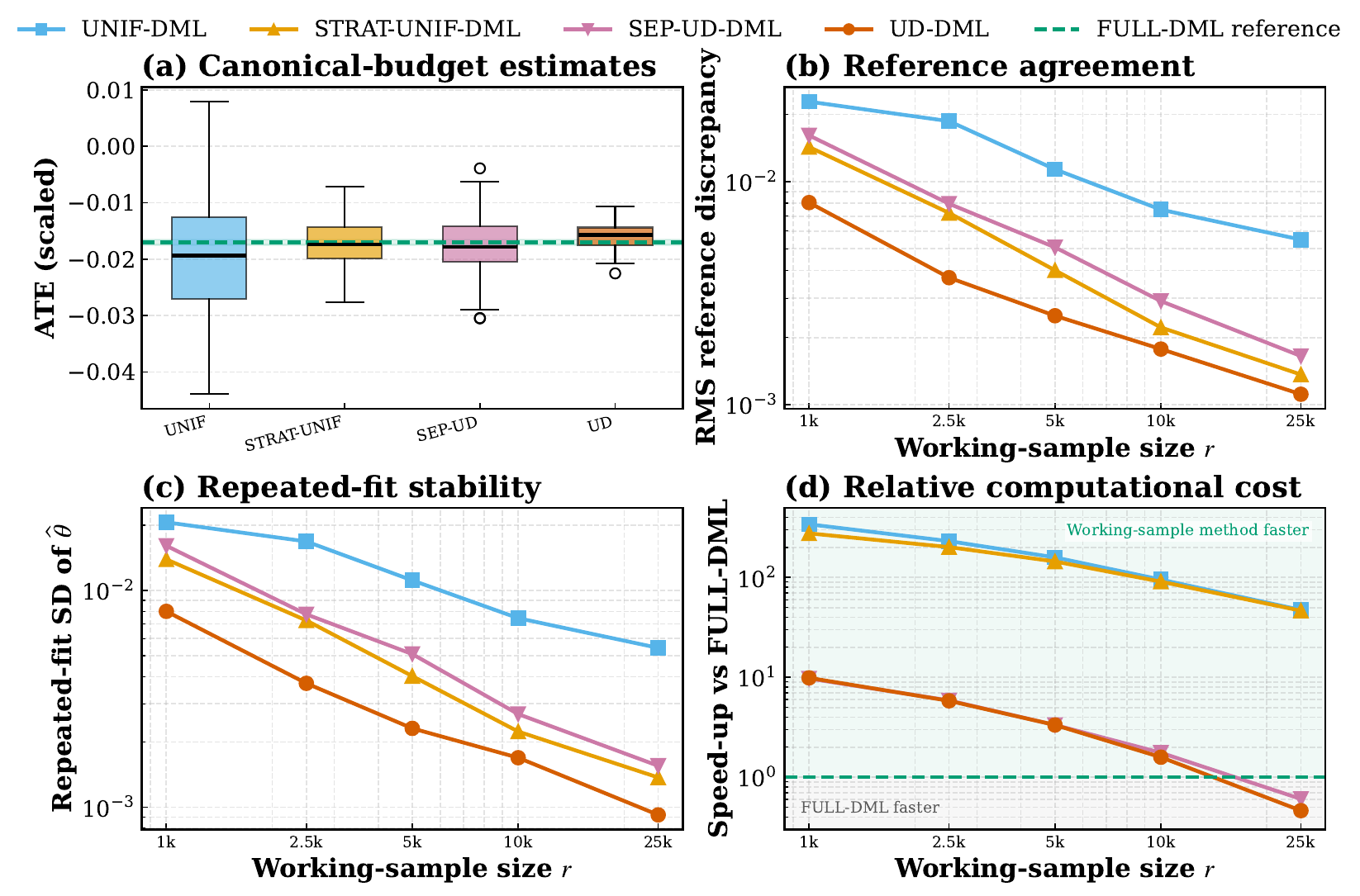}
\caption{Full fixed-data analysis of the 2021 US natality records ($n=2{,}846{,}543$, $B=100$ paired repetitions). Panel (a) shows the repeated-fit estimates for all four working-sample designs at the canonical budget $r=5000$; the dashed line and shaded band give the estimated FULL-DML reference and its 95\% interval. Panels (b) and (c) report RMS agreement with that reference and repeated-fit standard deviation over $r\in\{1000,2500,5000,10000,25000\}$. Panel (d) reports speed-up, defined as the one-off FULL-DML runtime divided by the mean method runtime; the parity line identifies the range of $r$ that yields a per-call computational saving.}
\label{fig:real_data_application}
\end{figure}

\clearpage

\end{document}